\ifx\synctex\undefined\else\synctex=1\fi

\documentclass[conference,compsoc]{IEEEtran}
\ifCLASSOPTIONcompsoc
  \usepackage[nocompress]{cite}
\else
  \usepackage{cite}
\fi
\ifCLASSINFOpdf
\else
\fi
\usepackage{amsmath}
\usepackage{stmaryrd}
\usepackage{url}
\usepackage{logic9-thm}

\usepackage{graphicx}
\usepackage{exscale}
\usepackage{tikz}
\usepackage{negotiations}
\usetikzlibrary{positioning,calc}
\usepackage{amsthm,thmtools,thm-restate}

\usepackage[textsize=small]{todonotes}

\newcommand{\SKIP}[1]{}

\newcommand{\Redl}{\mathit{Red}_\ell}
\newcommand{\Redn}{\mathit{Red}_n}
\newcommand{\tleq}{\sqsubseteq}

\newcommand{\bet}{\mathit{btw}}
\newcommand{\domleq}{\preceq}
\newcommand{\domle}{\prec}
\newcommand{\Nnn}{\Nn|_n}
\newcommand{\Nnl}{\Nn|_\ell}

\newcommand{\Red}{\mathit{Red}}

\newcommand{\bK}{\bar{K}}
\newcommand{\bG}{\bar{G}}

\newcommand{\init}{\mathit{init}}
\newcommand{\Cinit}{C_{\init}}
\newcommand{\Cfin}{C_{\fin}}

\newcommand{\N}{\mathcal N}

\newcommand{\dom}{\mathit{dom}}
\newcommand{\out}{\mathit{out}}
\newcommand{\Proc}{\mathit{Proc}}

\newcommand{\ignore}[1]{}

\newcommand{\first}[1]{I(#1)}
 \RequirePackage{amsthm}  
\newtheorem{theorem}{Theorem}
\newtheorem{definition}{Definition}

\newtheorem{lemma}{Lemma}
\newenvironment{example}[1][!*!,!]%
 {\noindent{\textbf{Example: }%
     \ifthenelse{\equal{#1}{!*!,!}}{}{%
       \normalfont\ (#1)}}}%
 {\koniec\medskip}
\renewenvironment{proof}[1][!*!,!]%
 {\noindent{\textbf{Proof: }%
     \ifthenelse{\equal{#1}{!*!,!}}{}{%
       \normalfont\ (#1)}}}%
 {\koniec\medskip}

\usepackage{algpseudocode,algorithmicx,algorithm}
\makeatletter
\newcommand{\algmargin}{\the\ALG@thistlm}
\makeatother

\renewcommand{\fin}{\mathit{fin}}

\renewcommand{\last}[1]{F(#1)}

\renewcommand{\act}[2][]{\xrightarrow[#1]{#2}}
\IEEEoverridecommandlockouts
\begin{document}
\title{Static Analysis of Deterministic Negotiations
\thanks{Accepted for publication in the Proceedings of LICS 2017}}

\author{\IEEEauthorblockN{Javier Esparza}
\IEEEauthorblockA{Technical University of Munich}
\thanks{Partially supported by the DFG Project \emph{Negotiations: A Model for
Tractable Concurrency}, and the Institute of Advance Studies of the Technische Universit\"at M\"unchen.}
\and
\IEEEauthorblockN{Anca Muscholl}
\IEEEauthorblockA{University of Bordeaux, LaBRI
% \\
% Email: anca@labri.fr
}
\and
\IEEEauthorblockN{Igor Walukiewicz}
\IEEEauthorblockA{CNRS, LaBRI, University of Bordeaux
%  \\
% Email: igw@labri.fr
}
}

\IEEEoverridecommandlockouts
\IEEEpubid{\makebox[\columnwidth]{978-1-5090-3018-7/17/\$31.00~
\copyright2017 IEEE \hfill} \hspace{\columnsep}\makebox[\columnwidth]{ }}
\maketitle

\begin{abstract}
Negotiation diagrams are a model of concurrent computation akin to workflow Petri nets. 
Deterministic negotiation diagrams, equivalent to the much
studied and used free-choice workflow Petri nets, are surprisingly
amenable to verification. Soundness (a property close to
deadlock-freedom) can be decided in \PTIME. Further, other fundamental questions
like computing summaries or the expected
cost, can also be solved in \PTIME\ for sound deterministic negotiation 
diagrams, while they are PSPACE-complete in the general case. 

In this paper we generalize and explain these results. We extend the
classical ``meet-over-all-paths'' (MOP) formulation of static analysis
problems to our concurrent setting, and introduce
Mazurkiewicz-invariant analysis problems, which encompass the questions
above and new ones.
% We show that the complexity of computing the MOP of an arbitrary
% Mazurkiewicz-invariant problem for a given sound deterministic
% negotiation is within a polynomial factor of the complexity of
% computing it for a sequential flow-graph of the same size, even though
% the state space of the negotiation can be exponentially larger.
We show that any
Mazurkiewicz-invariant analysis problem can be solved in \PTIME\ for sound deterministic
negotiations whenever it is in \PTIME\ for sequential flow-graphs---even though 
the flow-graph of a deterministic negotiation diagram can be exponentially
larger than the diagram itself.
This gives a common explanation to the low-complexity of all the
analysis questions studied so far.
Finally, we show that classical gen/kill analyses are also an instance of our
framework, and obtain a \PTIME\ algorithm for detecting anti-patterns in 
free-choice workflow Petri nets.

Our result is based on a novel decomposition theorem, of independent
interest, showing that sound deterministic negotiation diagrams can be
hierarchically decomposed into (possibly overlapping) smaller sound 
diagrams. 
\end{abstract}

\IEEEpeerreviewmaketitle

\section{Introduction }

Concurrent systems are difficult to analyze due to the state explosion
problem. Unlike for sequential systems, the flow graph of a concurrent system
is often exponential in the size of the system, so that analysis
techniques for sequential systems cannot be directly applied. 
One approach to analyze concurrent systems is to take a general model
and design heuristics that work for relevant examples.
Another, that we pursue in this paper, is to find a restricted class
of concurrent systems and design provably efficient algorithms for particular
analysis problems for this class.

In \cite{negI} Esparza and Desel introduced
\emph{negotiation diagrams}, a model of concurrency closely related to workflow Petri nets.
Workflow nets are a very successful formalism for the description of business processes, 
and a back-end for graphical notations like BPMN (Business Process Modeling Notation), 
EPC (Event-driven Process Chain), or UML Activity Diagrams (see e.g.\ \cite{aalst,van2004workflow}). 
In a nutshell, negotiation diagrams are workflow Petri nets that can be decomposed into communicating 
sequential Petri nets, a feature that makes them more 
amenable to theoretical study, while the translation into workflow nets 
(described in \cite{DBLP:conf/apn/DeselE15}) allows to transfer results and algorithms to 
business process applications. 

A negotiation diagram describes a distributed system with a fixed set of sequential processes. 
The diagram is composed of ``atomic negotiations'', each one involving a (possibly different) 
subset of processes. An atomic negotiation starts when all its participants are ready to engage in it, 
and concludes with the selection of one out of a fixed set of possible outcomes; for each 
participant process, the outcome determines which atomic negotiations the process
is willing to engage in at the next step. As workflow Petri nets, negotiations can simulate linearly bounded automata,
and so all interesting analysis problems are 
\PSPACE-hard for them. 

A negotiation is \emph{deterministic} if for every process the outcome
of an atomic negotiation completely determines the next atomic
negotiation the process should participate in. As shown in
\cite{DBLP:conf/apn/DeselE15}, the connection between negotiations
diagrams and workflow Petri nets is particularly tight in the deterministic case:
Deterministic negotiation diagrams are essentially isomorphic to the class of
\emph{free-choice} workflow nets, a class important in practice\footnote{For example, 70\% of the almost 2000
workflow nets from the suite of industrial models studied in
\cite{van2007verification,fahland2009instantaneous,DBLP:conf/fase/EsparzaH16}
are free-choice.} and extensively studied, see e.g. \cite{DBLP:conf/bpm/Aalst00,DBLP:journals/is/FavreFV15,DBLP:conf/bpm/FahlandV16,DBLP:conf/tacas/FavreVM16,DBLP:conf/fase/EsparzaH16,DBLP:conf/qest/EsparzaHS16,DBLP:conf/concur/EsparzaKMW16}). 
The state space of deterministic negotiations/free-choice workflow nets can
grow exponentially in their size, and so they are subject to the state
explosion problem. However, theoretical research has shown that,
remarkably, several fundamental problems can be solved in polynomial
time by means of algorithms that avoid direct
exploration of the state space (contrary to other techniques, 
like partial-order reduction, that only reduce the number of states to be explored, 
and still have exponential worst-case complexity). 
First, it can be checked in \PTIME\ if a deterministic negotiation diagram is
\emph{sound}~\cite{negII}, a variant of deadlock-freedom property~\cite{DBLP:conf/tacas/FavreVM16}\footnote{About 50\%
of the free-choice workflow nets from the suite mentioned above are sound.}.
Then, for sound deterministic negotiation diagrams \PTIME\ algorithms have been
proposed for: the  \emph{summarization problem}  \cite{DBLP:conf/fase/EsparzaH16},
 the problem of computing the \emph{expected cost} of a probabilistic
 free-choice workflow net \cite{DBLP:conf/qest/EsparzaHS16}, and the
identification of some \emph{anti-patterns} \cite{DBLP:conf/concur/EsparzaKMW16}. 

In this paper we  develop a generic approach to the static analysis of
sound deterministic negotiation diagrams. It covers all the problems above as particular instances,
and new ones, like the computation of the best-case/worst-case execution time.
The approach is a generalization to the concurrent setting of the classical lattice-based approach to 
static analysis of sequential flow-graphs \cite{nielson}.
A flow-graph consists of a set of nodes, modeling program points, and a set of edges,
modeling program instructions, like assignments or guards\footnote{In some papers the roles of nodes and edges are reversed: Nodes are program instructions, and edges are program points. The version with program points as nodes is more convenient for our purposes.} In the lattice-based approach one (i) defines a lattice ${\cal D}$ of dataflow informations capturing the analysis at hand, (ii) assigns semantic transformers $\sem{a} \colon {\cal D} \rightarrow {\cal D}$ to each action $a$ of the flow-graph, (iii) assigns to a path $a_1 \, \cdots \, a_n$ of the flow graph the functional composition $\sem{a_n} \circ \cdots \circ \sem{a_1}$ of the transformers, and (iv) defines the result of the analysis as the ``Merge Over all Paths'', i.e, the join of the transformers of all execution paths, usually called the MOP-solution or just the MOP of the dataflow problem. So performing an analysis amounts to computing the MOP of the flow-graph for the corresponding lattice and transformers.

Katoen {\em et al.} have recently shown in~\cite{DBLP:conf/acsd/Katoen12},~\cite{DBLP:conf/apn/EisentrautHK013} that in order to adequately deal with quantitative analyses of concurrent systems, like expected costs, one needs a 
semantics  that distinguishes between the inherent nondeterminism of each sequential process, and the nondeterminism introduced by concurrency (the choice of the process that should perform the next step). Following these ideas, we introduce a semantics in which the latter is resolved by an external scheduler, and define the MOP for a given scheduler. The result of a dataflow analysis is then 
given by the infimum or supremum, depending on the application, of the MOPs for
all possible schedulers.

The contributions of the paper are the following:

(1) We present an extension of a static analysis framework to deterministic
negotiation diagrams. 
In particular, we identify the class of \emph{Mazurkiewicz invariant
frameworks} that respect the concurrency relation in negotiations. 
We prove a theorem showing a first important property of sound
deterministic negotiations, namely that the MOP is independent of the
scheduler for Mazurkiewicz invariant frameworks. 
This allows to compute the result of the analysis by fixing a
scheduler, and computing the MOP for it. 
As an another motivation for Mazurkiewicz invariant frameworks we
observe that there are static analysis frameworks for which analysis is
\NP\ hard, even for sound deterministic negotiation diagrams.
%However, we still have to solve the problem of computing this MOP. 

(2) The main contribution of the paper is a method to compute MOP
problems for sound deterministic negotiation diagrams.
The method does not require the computation of the reachable configurations.
We prove a novel \emph{decomposition theorem} showing that a
deterministic negotiation diagram is composed of smaller subnegotiations
involving only a subset of the processes, and that these subnegotiations
are themselves sound. This allows us to define a generic \PTIME\ 
algorithm for computing the MOP for Mazurkiewicz invariant static
analysis frameworks.

(3) Finally, we show that the problems studied in
\cite{DBLP:conf/fase/EsparzaH16,DBLP:conf/qest/EsparzaHS16}, and
others, are Mazurkiewicz invariant. Further, we show that the MOP of
an important class of analyses --  all four flavors of gen/kill problems,
well known in the static analysis community -- can be reformulated as invariant
frameworks, and computed in \PTIME.

\emph{Organization of the paper:} Section \ref{sec:def} introduces the negotiation 
model and static analysis frameworks. Section \ref{sec:decomp} proves the decomposition 
theorem. Section \ref{sec:compMOP} presents the algorithm to compute the MOP of
an arbitrary Mazurkiewicz-invariant analysis framework. Section
\ref{sec:genkill} deals with gen/kill analyses.

%All missing proofs can be found in \cite{}.

% It is well-known that, when the transformers are monotonic, the  MOP can be approximated 
% by computing the least fixed point of a system $X = f(X)$ of linear equations over the carrier $C$ of the lattice, where the dimension of the vector $X$ is equal to the number of nodes of the flow-graph.
% Moreover, if the transformers are also distributive, then the MOP is
% equal to the least fixed point, which, for historical reasons, is
% usually called the ``Maximal Fixed Point'' or MFP, and so computing
% the MOP reduces to computing the MFP.  
% However, for negotiations this approach is prohibitively inefficient,
% even if they are deterministic. 
% If one directly applies this approach to
% deterministic negotiations,  one has to compute the MOP for a scheduler as the
% MFP of a system of equations containing one equation for every
% reachable configurations of the negotiation. 
% This suffers from the state explosion problem as the number of reachable
% configurations can be exponential in the size of the negotiation
% diagram. 

\medskip

\noindent \textbf{Related work.} 
As we have mentioned, deterministic negotiations
are very close to free-choice workflow Petri nets, also called
workflow graphs. Algorithms for the
analysis of specific properties of these nets have been studied in
\cite{DBLP:conf/bpm/Aalst00,DBLP:conf/caise/TrckaAS09,DBLP:journals/is/FavreFV15,DBLP:conf/bpm/FahlandV16,DBLP:conf/tacas/FavreVM16,DBLP:conf/fase/EsparzaH16,DBLP:conf/qest/EsparzaHS16,DBLP:conf/concur/EsparzaKMW16}. 
We  have already described above the relation to these works.

We discuss the connection to work on static analysis for (abstract models of) programming languages.
The synchronization-sensitive analysis of concurrent programs has been intensively studied
(see e.g \cite{DBLP:conf/popl/EsparzaP00,DBLP:journals/toplas/Ramalingam00,DBLP:conf/esop/SeidlS00,DBLP:conf/pldi/ChughVJL08,
DBLP:conf/tacas/FarzanM07,DBLP:conf/sas/LammichM08,DBLP:conf/sas/FarzanK10,
DBLP:conf/popl/SchwarzSVLM11,DBLP:journals/toplas/GantyM12,DBLP:journals/toplas/BouajjaniE13}). A
fundamental fact is that interprocedural synchronization-sensitive analysis is undecidable \cite{DBLP:journals/toplas/Ramalingam00}, and intraprocedural synchronization-sensitive analysis has high complexity (ranging from PSPACE-completeness to EXPSPACE-completeness, depending on the communication primitive, see e.g. \cite{DBLP:journals/acta/ReifS88}). This is in sharp contrast to the linear complexity of static analysis in the size of the flow graph for sequential programs, and causes work on the subject to roughly split into two research directions. The first one aims at obtaining decidability or low complexity of analyses by restricting the possible synchronization patterns. Many different restrictions have been considered: parbegin-parend constructs \cite{DBLP:conf/popl/EsparzaP00,DBLP:conf/esop/SeidlS00}, generalizations thereof  (see e.g.~\cite{DBLP:conf/sas/LammichM08}),  synchronization by nested locks (see e.g.~\cite{DBLP:conf/sas/FarzanK10}), and asynchronous programming (see e.g.\cite{DBLP:journals/toplas/GantyM12}).
The other direction does not restrict the synchronization patterns, at the price of worst-case exponential analysis algorithms (see e.g. \cite{DBLP:conf/tacas/FarzanM07,DBLP:conf/pldi/ChughVJL08}, where control-flow of parallel programs is modelled by Petri nets, and a notion similar to Mazurkiewicz invariance is also used). 

Compared with these papers,  the original feature of our work is that
we obtain polynomial analysis algorithms without restricting the
possible synchronization patterns; instead, deterministic negotiation diagrams
restrict the {\em interplay} between synchronization and choice. This distinction can be best appreciated
when we compare these formalisms, but excluding choice. In the programming languages of
\cite{DBLP:conf/popl/EsparzaP00,DBLP:conf/esop/SeidlS00,DBLP:conf/sas/LammichM08,DBLP:conf/sas/FarzanK10,DBLP:journals/toplas/GantyM12}, excluding choice means excluding if-then-else or alternative constructs, 
while for deterministic negotiations it means considering the special
case in which every node has exactly one outcome. 
Sound deterministic negotiation diagrams can model all synchronization
patterns given in terms of Mazurkiewicz traces, but it is not the case
for formalisms of \cite{DBLP:conf/popl/EsparzaP00,DBLP:conf/esop/SeidlS00,DBLP:conf/sas/LammichM08,DBLP:conf/sas/FarzanK10,DBLP:journals/toplas/GantyM12}.
% It is then easy to
% construct synchronization patterns that can be  modeled by a simple deterministic negotiation, but not by the formalisms of
% \cite{DBLP:conf/popl/EsparzaP00,DBLP:conf/esop/SeidlS00,DBLP:conf/sas/LammichM08,DBLP:conf/sas/FarzanK10,DBLP:journals/toplas/GantyM12}.
 For example, the languages of
\cite{DBLP:conf/popl/EsparzaP00,DBLP:conf/esop/SeidlS00} cannot model
a synchronization pattern with three processes $A, B, C$ in which first $A$ 
synchronizes with $B$, then $A$ synchronizes with $C$, and finally $B$ synchronizes with
$C$. Observe that on the other hand, negotiations are finite state,
whereas the other formalisms we have mentioned have non-determinism,
recursion, and possibly, thread creation.

%%% Local Variables:
%%% mode: latex
%%% TeX-master: "../mlics"
%%% End:

\section{Negotiations}\label{sec:def}

A~\emph{negotiation diagram} $\Nn$ is a tuple
$\struct{\Proc,N,\dom,R,\d}$, where $\Proc$ is a finite
set of \emph{processes} (or agents) 
and $N$ is a finite set of \emph{nodes} where
the processes can synchronize to choose an \emph{outcome}.
The function $\dom:N\to \Pp(\Proc)$ associates to every node $n \in N$ 
the (non-empty) set $\dom(n)$ of processes participating in it. Nodes
are denoted as $m$ or $n$, and processes as $p$ or $q$; possibly with
indices. The set of possible outcomes of nodes is denoted $R$\footnote{$R$ stands for \emph{result};
we prefer to avoid the confusing symbol $O$.},
and we use $a,b,\ldots$ to range over its elements. Every node $n\in N$ has its set
of possible outcomes $\out(n) \subseteq R$.

The control flow in a negotiation diagram is determined by a partial transition function
$\d:N\times R\times P\act{\cdot} \Pp(N)$, telling that after the
outcome $a$ of node $n$, process $p\in\dom(n)$ is
ready to participate in any of the nodes in the set
$\d(n,a,p)$.  So for every $n'\in\d(n,a,p)$ we have $p\in
\dom(n')\cap\dom(n)$, and for every $n$, $a\in \out(n)$ and
$p\in\dom(n)$ the result $\d(n,a,p)$ is defined. Observe that nodes may
have one single participant process, and/or have one single
outcome. A \emph{location} is a pair $(n,a)$ such that $a \in \out(n)$,
and we define its domain as $\dom(n)$.

\begin{figure}[ht]
\centerline{\scalebox{0.70}{\begin{tikzpicture}
\vnego[ports=2,id=n0,spacing=2.0]{0,0}
\node[above left = 0.5cm and -0.1 cm of n0_P0, font=\large] {$n_0$};
\node[left = 0.5cm of n0_P1, font=\large] {$D_1$};
\node[left = 0.5cm of n0_P0, font=\large] {$D_2$};
\vnego[ports=1,id=n1]{2.0,2.0}
\node[above = 0.2cm  of n1_P0, font=\large] {$n_1$};
\vnego[ports=1,id=n2]{2.0,0}
\node[below  = 0.2cm of n2_P0, font=\large] {$n_2$};
\vnego[ports=1,id=n3]{4.0,2.0}
\node[above = 0.2cm of n3_P0, font=\large] {$n_3$};
\vnego[ports=2,id=n4,spacing=2.0]{6,0}
\node[above left = 0.5cm and -0.1 cm of n4_P0, font=\large] {$n_4$};
\vnego[ports=1,id=n5]{7.0,-1.0}
\node[below = 0.2cm of n5_P0, font=\large] {$n_5$};
\vnego[ports=1,id=n6]{9.0,-1.0}
\node[below = 0.2 cm of n6_P0, font=\large] {$n_6$};
\vnego[ports=2,id=n7,spacing=2.0]{10,0}
\node[above left = 0.5cm and -0.1 cm of n7_P0, font=\large] {$n_7$};

\pgfsetarrowsend{latex}
\draw (n0_P0) -- (n2_P0) node [above,midway]{{\it reg}};
\draw (n0_P1) -- (n1_P0) node [above,midway]{{\it reg}};
\draw (n1_P0) -- (n3_P0) node [above,midway]{{\it send}};
\draw (n2_P0) -- (n4_P0) node [above,midway]{{\it eval}};
\draw (n3_P0) to [bend left = 20] node [above,midway]{{\it tout}} (n4_P1);
\draw (n3_P0) to [bend right = 20] node [below,midway]{{\it rec}} (n4_P1);
\draw (n4_P1) to [bend left = 15] node [above,midway]{{\it npr}} (n7_P1);
\draw (n4_P1) to [bend right = 15] node [below,midway]{{\it pr}} (n7_P1);
\draw (n4_P0) -- (n7_P0) node [above,midway]{{\it npr}};
\draw (n4_P0) -- (n5_P0) node [left,near end]{{\it pr}};
\draw (n5_P0) to [bend left = 20] node [above,midway]{{\it done}} (n6_P0);
\draw (n6_P0) to [bend left = 20] node [below,midway]{{\it nOK}} (n5_P0);
\draw (n6_P0) -- (n7_P0) node [right,near start]{{\it OK}};
%\draw (1,0) -- (n2_P0);
%\draw (n0_P1) -- (1,2) -- (n2_P1);
%\draw (1,2) -- (n3_P0);
%\draw (n1_P0) -- (n4_P0);
%\draw (n2_P0) -- (n4_P0);
%\draw (n2_P1) -- (n4_P1);
%\draw (n3_P0) -- (n4_P1);
\end{tikzpicture}}}
\caption{\small{A negotiation diagram with two processes.}}
\label{fig:example1}
\end{figure}
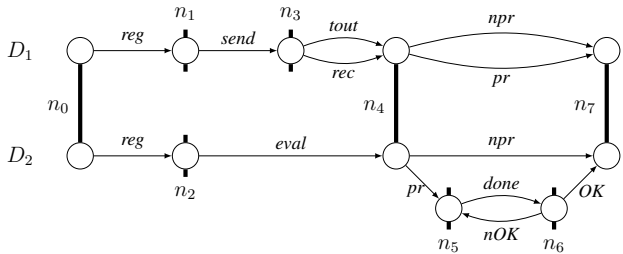

% Negotiations admit a graphical representation. Figure~\ref{fig:example1}
% shows a negotiation with $\Proc = \{ p, q \}$,
% $N = \{n_0, \ldots, n_7\}$ and $R$ is the set of labels on the arrows. For example, we have
% $\dom(n_1)=\{p,q\}$, $\d(n_1,b,p)= \{n_3\}$ and $\d(n_1,b,q)= \{n_6\}$.
% More details can be found in \cite{negI}.

\begin{example}
Figure \ref{fig:example1} shows a negotiation diagram for (a slight
modification of) the well-known insurance claim example of
\cite{aalst} (see also Fig. 2 of \cite{DBLP:conf/fase/EsparzaH16} for
a workflow Petri net model). The diagram describes a workflow for
handling insurance claims by an insurance company with two departments
$D_1$ and $D_2$. The processes of the negotiation are $D_1$ and
$D_2$. The nodes $n_0, n_4, n_7$ have domain $\{D_1, D_2\}$; $n_1$ and
$n_3$ have domain $D_1$, and $n_2, n_3, n_6$ have domain $D_2$. After
the claim is {\it reg}istered, outcome \emph{reg} involves both
processes, $D_1$ {\it send}s a questionnaire to the client, and concurrently $D_2$ makes a first {\it eval}uation of the claim. After the client's answer is {\it rec}eived or a time-out occurs
(outcome {\it tout}), both departments decide together at node
$n_4$ whether to {\it pr}ocess the claim or not. In both cases $D_1$ has nothing further to do, and moves to the final node $n_7$. If the decision is to process the claim, then $D_2$ moves to $n_5$, and the claim is processed, possibly several times, until a satisfactory result is achieved (outcome {\it OK}), after which $D_2$ also moves to $n_7$. 
\end{example}

A \emph{configuration} of a negotiation diagram is a function $C:\Proc\to
\Pp(N)$ mapping each process $p$ to the set of nodes in which
$p$ is ready to engage. A node $n$ is \emph{enabled} in a 
configuration $C$ if $n\in C(p)$ for every $p\in \dom(n)$, that is, if all 
processes that participate in $n$ are ready to proceed with it. A configuration 
is a \emph{deadlock} if it has no enabled node. If node $n$ is enabled in $C$, and $a$ is an outcome
of $n$, then we say that location $(n,a)$ can be \emph{executed}, and its 
execution produces a new configuration $C'$ given by $C'(p)=\d(n,a,p)$ 
for $p\in\dom(n)$ and $C'(p)=C(p)$ for $p\not \in \dom(n)$.  We denote this by
$C\act{(n,a)}C'$. For example, in Figure~\ref{fig:example1} we have 
$C\act{(n_0,{\it reg})}C'$ for $C(D_1) = \{n_0\}= C(D_2)$ and
$C'(D_1) = \{n_3\}, C'(D_2)=\{n_2\}$.

A~\emph{run} of a negotiation diagram
$\Nn$ from a configuration $C_1$ is a finite or infinite sequence
of locations $w=(n_1,a_1)(n_2,a_2)\dots$ such that there are configurations
$C_2,C_3,\dots$ with 
\begin{equation*}
  C_1\act{(n_1,a_1)} C_2\act{(n_2,a_2)} C_3\dots
\end{equation*}
We denote this by $C_1\act{w}$, or $C_1\act{w}C_k$ if the sequence
is finite and finishes with $C_k$.
In the latter case we say that $C_k$ is \emph{reachable from $C_1$ on $w$}.
We simply call it \emph{reachable} if $w$ is irrelevant, and write
$C_1\act{*}C_k$. 

Negotiation diagrams come equipped with two distinguished \emph{initial} and
\emph{final} nodes $n_\init$ and $n_\fin$, in which {\em all} processes in
$\Proc$ participate. The \emph{initial and final
configurations} $\Cinit$, $\Cfin$ are given by $\Cinit(p)=\set{n_\init}$
and $\Cfin(p)=\set{n_\fin}$ for all $p\in \Proc$. 
A run is \emph{successful} if it starts in $\Cinit$ and ends in
$\Cfin$.  We assume that every node  (except  for
$n_\fin$) has at least one outcome. In Figure~\ref{fig:example1},
$n_\init = n_0$ and $n_\fin = n_7$.

A negotiation diagram $\Nn$ is \emph{sound} if every partial run 
starting at $\Cinit$ can be completed to a successful run. 
If a negotiation diagram has no infinite runs, then it is sound if{}f it has no
reachable deadlock configuration. The negotiation diagram of Figure~\ref{fig:example1}
is sound.

Process $p$ is \emph{deterministic} in a negotiation diagram $\Nn$ if for
every $n\in N$, and $a\in R$, the set $\d(n,a,p)$ of possible successor nodes
is either a singleton or the empty set.
A negotiation diagram is \emph{deterministic} if every process $p\in\Proc$ is
deterministic. The negotiation diagram of Figure~\ref{fig:example1} is deterministic.

The \emph{graph of a negotiation diagram} has $N$ as set of vertices, 
and there is an edge $n\act{p,a} n'$ if{}f $n'\in \d(n,a,p)$. Observe
that $p\in \dom(n)\cap\dom(n')$.

A negotiation diagram is \emph{acyclic} if its graph is
acyclic. Acyclic negotiation diagrams
cannot have infinite runs, so as mentioned above,
soundness is equivalent to deadlock-freedom. 
%For an acyclic negotiation $\Nn$ we fix a linear order $\leqN$ on its nodes
%that is a topological order on the graph of $\Nn$. 
%This means that if there is an edge from $m$ to $n$ in the graph of
%$\Nn$ then $m\leqN n$. 

%The~\emph{restriction} of a negotiation $\Nn$ to a subset of its processes
%$\Proc'$ is the negotiation
%$\struct{\Proc',N',\dom',R,\d'}$ where
%$N'$ is the set of those $n\in N$
%for which $\dom(n)\cap\Proc'\not=\es$, $\dom'(n)=\dom(n)\cap\Proc'$,
%and $\d'(n,r,p)=\d(n,r,p)\cap N'$. The restriction of $\Nn$ to
%deterministic processes is denoted as $\Nn_D$ throughout the paper.

% A negotiation $\Nn$ is \emph{det-acyclic} if $\Nn_D$ is acyclic.
% It follows easily from the definitions that a weakly non-deterministic, 
% det-acyclic negotiation does not have any infinite run. 

\subsection{Static analysis frameworks}
\label{subsec:frameworks}

Let $(D, \sqcup, \sqcap, \sqsubseteq, \bot, \top)$ be a complete lattice.
A function $f \colon D \rightarrow D$ is {\em monotonic} if $d \sqsubseteq d'$ implies $f(d) \sqsubseteq f(d')$, and distributive if $f(\bigsqcap D') = \bigsqcap \{ f(d) \mid d \in D'\}$. 

An {\em analysis framework}\footnote{In \cite{nielson} this is called
  a monotone and distributive framework.} of a negotiation diagram is
a lattice together with a mapping $\sem{\_}$ that assigns to each
outcome $\ell$ a monotonic and distributive function $\sem{\ell}$ in
the lattice. Abusing language, we use $\sem{\_}$ to denote a framework.

A negotiation diagram has two kinds of nondeterminism, one that picks a node among the ones enabled at a configuration, and a second kind which picks an outcome. We distinguish the two by letting a scheduler to decide the first kind. This is an important design choice, motivated
by modeling issues: In distributed systems, one often has information about how the outcomes are picked, but not about the way nondeterminism due to concurrency is resolved. In particular, one may have probabilistic information about the former, but not about the latter. This point has been discussed in detail by Katoen {\it et al.\!} \cite{DBLP:conf/acsd/Katoen12},~\cite{DBLP:conf/apn/EisentrautHK013}, who also advocate the separation of the two kinds of nondeterminism.

A {\em scheduler} of $\Nn$ is a partial 
function $S$ that assigns to every run $\Cinit \act{w} C$ a node $S(w)$ enabled at $C$, if it exists. 
A finite initial run $w= \ell_1 \cdots \ell_k$, where $\ell_i = (n_i, a_i)$,
is {\em compatible} with $S$ if $S(\ell_1 \cdots \ell_i)=
n_{i+1}$ for every $1 \leq i \leq k-1$.

For example, a scheduler for the negotiation diagram in Figure~\ref{fig:example1} can give preference to $n_1$ and $n_3$ over $n_2$. The successful runs compatible with this scheduler are given by the regular expression (omitting the nodes of the locations)
$\mathit{reg} \, \mathit{send} \left(\mathit{tout}|\mathit{rec}\right) \mathit{eval} 
\left(\mathit{npr}|\mathit{pr} (\mathit{done} \, \mathit{nOK})^*\mathit{done} \mathit{OK}\right)$.

The abstract semantics of a finite  run $w = \ell_1 \cdots \ell_k$ is 
the function $\sem{w}:= \sem{\ell_k} \circ \sem{\ell_{k-1}} \circ  \cdots \circ \sem{\ell_1}$. 
The abstract semantics of $\Nn$ with respect to a scheduler $S$ is the function $\sem{\Nn,S}$ defined by
\begin{eqnarray*}
  \sem{\Nn,S} &=& \bigsqcup \; \big\{\; \sem{w} \mid \text{$w$ is a
  successful run}\\
&& \text{ of $\Nn$ compatible with $S$} \big\}
\end{eqnarray*}
where the extension of $\sqcup$ to functions is defined pointwise.

The \emph{abstract semantics} $\sem{\Nn}$ of $\Nn$ is defined as
either $\bigsqcap  \big\{ \, \sem{\Nn,S} \mid \mbox{$S$ is a scheduler
  of $\Nn$} \, \big\}$, or as $\bigsqcup \big\{ \, \sem{\Nn,S} \mid
\mbox{$S$ is a scheduler of $\Nn$} \, \big\}$, depending on the application.

In classical static analysis, analysis frameworks are over
flow-graphs, instead of negotiation diagrams \cite{nielson}. Flow-graphs
describe sequential programs. Loosely speaking, a flow-graph is a
graph whose nodes are labeled with program points, and whose edges are
labeled with program instructions (assignments or guards). The mapping
$\sem{\_}$ assigns to an edge the relation describing the effect of
the assignment or guard on the program variables.  We can see a
flow-graph as a degenerate negotiation diagram in which all nodes have one single process. In this case every reachable
configuration enables at most one node, and so there is a unique
scheduler. 
So, in this case, the abstract semantics of a flow-graph is the
standard ``Merge Over all Paths'' (MOP), defined by $\sem{\Nn} = \bigsqcup
\big\{\, \sem{w} \mid \mbox{ $w$ is a successful run}\, \}
$.\footnote{Some classical
literature uses $\bigsqcap$ instead of $\bigsqcup$ and speaks of the
``Meet Over all Paths'', but other standard texts, e.g. \cite{nielson}, use $\bigsqcup$. }

Several interesting
analyses are instances of our framework. \label{sec:examples}

\subsubsection{Input/output semantics} Let $V$ be a set of variables
and $Z$ the set of values.
A {\em valuation} is a function $V\rightarrow Z$, and ${\it Val}$ denotes the 
set of all valuations. An element of $D$ is a set $d \subseteq {\it Val}$. The join and meet lattice operations are set union and intersection. 
For each location $\ell=(n,a)$, the function $\sem{\ell}$ describes for each input valuation 
$v \in {\it Val}$ the set of output valuations $\sem{\ell}(\{v\})$ that are possible if $n$ ends with 
result $a$. For any 
set of valuations the function is defined by $\sem{\ell}(V) = \bigcup_{v \in V} \sem{\ell}(\{v\})$.
The semantics $\sem{\Nn}$ is the relation that assigns to every initial valuation
the possible final valuations after a successful run. 
%So $\sem{\Nn}$ is the solution to the
%\emph{summarization problem} studied in \cite{DBLP:conf/fase/EsparzaH16}.

\subsubsection{Detection of anti-patterns}
Actions in business processes generate, use, modify, and delete resources
(for example, a document can be created by a first department,
read and used by a second, and classified as confidential by a third).
Anti-patterns are used to describe runs that do not correctly access resources; for example,
a resource is used before it is created, or a resource is created and then never used.
Examples of anti-patterns can be found in \cite{DBLP:conf/caise/TrckaAS09}. They can be easily formalized as analyses frameworks. Consider for example two locations $\ell_1$ and $\ell_2$ that
generate a resource, and a set $K$ of locations that delete it. We
wish to know if a given deterministic sound negotiation diagram has a
successful run that belongs to  
\begin{equation*}
\label{eq:eq1}
L=\Ll^*\ell_1(\bK)^*\ell_2\Ll^*  
\end{equation*}
\noindent where $\bK$ denotes the set of locations not in $K$. 
In other words, is there a scenario where a resource is generated twice
without deleting it in between.

To encode this problem in our static analysis framework, 
we take $D=\set{0,1,2}$ with the natural order together with $\min,\max$ as
$\sqcap$ and $\sqcup$, respectively. 
Intuitively, $0$ says that the sequence does not have a suffix of the
form $\ell_1(\bK)^*$, $1$ says that it has such a suffix, and $2$ that
it has a subword $\ell_1(\bK)^*\ell_2$.
The semantics of a location is a monotone and distributive function
from $D$ to $D$ reflecting this intuition:

\begin{align*}
  \sem{\ell_1}(x)=&\begin{cases}
    2 & \text{if $x=2$}\\
    1 & \text{otherwise}
  \end{cases}
  \hspace{0.3cm}
  \sem{\ell_2}(x)=\hspace{-0.3cm}&
                  \begin{cases}
                    0 & \text{if $x=0$}\\
                    2 & \text{otherwise}
                  \end{cases}
\end{align*}
\begin{align*}
  \sem{\ell}(x)=&\begin{cases}
    x & \text{if $\ell\in \bK$}\\
    2 & \text{if $\ell\in K$ and $x=2$}\\
    0 & \text{if $\ell\in K$ and $x=0,1$}
  \end{cases}
\end{align*}

\subsubsection{Minimal/maximal  expected cost} We let $D = \{ (p, c) \mid p \in \mathbb{R}^+_0, c \in \mathbb{R} \}$, 
where we interpret $p$ as a probability and $c$ as a cost. We take
$(p_1, c_1) \sqcup (p_2, c_2) = (p_1+p_2, c_1 p_1 + c_2
p_2)$ and $(p_1, c_1) \sqsubseteq (p_2, c_2)$ if $p_1 \leq p_2$ and $c_1 \leq c_2$.

We define a function ${\it Prob} \colon N \times R \rightarrow [0,1]$ such that ${\it Prob}(n, a)=0$ if
$a \notin {\it out}(n)$, and $\sum_{a \in R} {\it Prob}(n, a) =1$ for every $n \in N$. Intuitively, 
${\it Prob}(n, a)$ is the probability that node $n$ yields the outcome $a$. We also define
a cost function ${\it Cost} \colon N \times R \rightarrow \mathbb{R}$ that assigns to each result a cost.

Let $\sem{\ell}( (p, c) ) = \big( p \cdot {\it Prob}(\ell), c + {\it Cost}(\ell) \big)$. Then $\sem{\Nn,S}(1, 0)$ gives the expected cost of $\Nn$ under the scheduler $S$ (which may be infinite)
and $\sem{\Nn}(1,0)$ is the minimal/maximal expected cost. 

\subsubsection{Best/worst-case execution time} 
Let $\mathbb{R}^+_0$ denote the nonnegative reals. A {\em time valuation} is a function 
$v \colon {\it Proc} \rightarrow \mathbb{R}^+_0 \cup \{ \infty \}$ that assigns
 to each process $p$ a time $v(p)$, intuitively corresponding to the time that the 
process has needed so far. The elements of $D$ are time valuations, 
with $(v \sqcup v')(p) = \max\{ v(p),v'(p) \}$ for 
every process $p$, and $v \sqsubseteq v'$ if $v(p) \leq v'(p)$ for every process $p$, 

We assign to each outcome $\ell=(n,a)$ and to each process $p \in \dom(n)$  the time $t_{\ell,p}$ that 
$p$ needs to execute $a$. The semantic function $\sem{\ell}$ is given by $\sem{\ell}(v) = v'$, where 
$$v'(p) =  \left\{ 
\begin{array}{cl}
v(p) &  \mbox{ if $p \notin \dom(n)$ } \\ 
\displaystyle\max_{p' \in \dom(n)} v(p') + t_{\ell,p} &  \mbox{ if $p \in \dom(n)$ }  
\end{array} 
\right.$$

This definition reflects that all processes in $\dom(n)$ must wait until all of them are ready, and  
then we add to them the time they need to execute $\ell$. Since the initial and final atoms involve all processes,  
the abstract semantics $\sem{w}$ of a successful run has the form  
$\sem{w}(v) = (\max_{p \in \Proc} v(p)  + t_w(p))_{p \in\Proc}$, where $t_w(p)$ is  
the time process $p$ needs to execute $w$. In particular,
we have $\sem{w}(0) = t_w$. 
Then $\sem{\Nn,S}(0)$ gives the best-case execution time for a scheduler $S$, and $\sem{\Nn}(0)$
the infimum/ supremum over the times for each scheduler.

\subsection{Maximal fixed point of an analysis framework}

It is well-known that for sequential flow-graphs the MOP of an analysis framework coincides with the \emph{Maximal Fixed Point} of the framework, or \emph{MFP}. The MOP is the least fixed point of a set of linear 
equations over the lattice, having one equation for each node of the flow-graph\footnote{Again, the name ``maximal'' has historical reasons.}. The least fixed point can 
be approximated by means of Kleene's theorem, and computed exactly in a number of cases, including 
the case of lattices satisfying the ascending chain condition, but also others. For example, the lattice
for the expected cost of a flow-graph does not satisfy the ascending chain condition, but yields a set
of linear fixed point equations over the rational numbers, which can be solved using standard techniques.

In the concurrent case, the correspondence between MOP and MFP is more
delicate. Given a scheduler $S$, we can construct the reachability
graph of the negotiation diagram, corresponding to the runs compatible with $S$. If the graph is finite (for instance, this is always the case if the scheduler is memoryless, i.e., the node selected by the scheduler 
to extend a run depends only on the configuration reached by the run), then $\sem{\Nn,S}$ can be computed as the MFP of this graph, seen as a sequential flow-graph. The corresponding set of linear fixed point equations has one equation for each configuration of the graph. However, this approach has two problems: 
\begin{itemize}
\item[(a)] The number of schedulers is infinite, and non-memoryless schedulers may generate an infinite reachability graph;
so we do not obtain an algorithm for computing $\sem{\Nn}$.
\item[(b)] Even for memoryless schedulers, the size of the reachability graph may grow exponentially in the size of the negotiation diagram. So the algorithm for computing $\sem{\Nn,S}$ needs exponential time, also for lattices with only two elements.
\end{itemize}

In the remaining of this section we introduce a 
\emph{Mazurkiewicz-invariant analysis framework},
and show that for this framework and for the class of sound
deterministic negotiation diagrams
we can overcome these two obstacles. In Section \ref{subsec:mazur} we solve problem (a): We show that 
$\sem{\Nn} = \sem{\Nn,S}$ for every scheduler $S$ (Theorem \ref{thm:mazur} below), and so that it 
suffices to compute $\sem{\Nn,S}$ for a scheduler 
$S$ of our choice. 
In the rest of the paper we solve problem (b): We give a 
procedure that computes $\sem{\Nn}$ without ever constructing the
reachability graph of the negotiation diagram. The procedure reduces the
problem of computing the MOP to computing the MFP of a polynomial
number of (sequential) flow-graphs, each of them of size at most
linear in the size of the negotiation diagram. This shows that the MOP
can be computed in polynomial time for a sound deterministic
negotiation diagram if{}f it can be computed in polynomial time for a sequential flow-graph.

If we remove any of the three conditions of our setting
(Mazurkiewicz-invariance, soundness, determinism), then there exist
frameworks with the following property: deciding if the MOP has a
given value is polynomial in the sequential case (i.e., for flow
graphs of sequential programs), but  at least NP-hard for
negotiations. 

We sketch the NP-hardness proof for deterministic, sound
negotiations where the framework is not
Mazurkiewicz-invariant. Consider the NP-hard problem 1-in-3-SAT, where
it is asked if for a CNF formula with $k$ variables and $m$ clauses
there is an assignment that sets exactly one literal true in each
clause. We have $k$ processes $p_1,\ldots,p_k$, one for each variable
$x_i$. We describe the (acyclic) deterministic, sound negotiation
$\Nn$. 
The initial node of $\Nn$ has a single outcome, that leads process $p_i$
to a node with domain 
$\set{p_i}$. From there $p_i$ branches for the two possible
values for $x_i$. The ``true'' branch is a line with outcomes
corresponding to clauses that become ``true'' when $x_i$ is true, and
analogously for the ``false'' branch - in both cases respecting the
order of clauses. Let us denote by $C_j$ an outcome corresponding to
the $j$-th clause. The lattice $D$ has elements $\bot < 1,  \dots,
(m+1) < \top$; so there are $m+1$ pairwise incomparable elements
together with $\bot$ and $\top$. 
For every node $n$ and clause $C_j$ we set $\sem{(n,C_j)}(j)=j+1$,
$\sem{(n,C_j)}(\top)=\top$ and 
$\sem{(n,C_j)}(d)=\bot$ otherwise. 
Moreover $\sem{\ell}$ is the identity function for all other locations $\ell$.
This framework is monotonic and distributive.
For a run $w$ of $\Nn$ we have  
$\sem{w}(1)=m+1$ if the subsequence of clauses appearing in $w$ is exactly
$C_1\dots C_m$; otherwise $\sem{w}(1)=\bot$.
Since $\sem{\Nn}$ is the $\bigsqcup$ over all runs, we get that
the 1-in-3-SAT instance is positive iff $\sem{\Nn}(1)=m+1$.
In the sequential case, the analysis can be done in polynomial time,
since the lattice $D\to D$ has the height $\Oo(m)$.

The proofs of the other two cases (where determinism or soundness are removed)
follow easily from a simple construction shown in Theorem 1 of
\cite{negI}: Given a deterministic linearly 
bounded automaton $A$ and a word $w$, one can construct in polynomial
time a negotiation 
$\Nn_A$ having one single run that simulates the execution of $A$ on
the input $w$. 
This gives PSPACE-hardness for essentially all non-trivial frameworks,
Mazurkiewicz invariant or not.

% PSPACE-hard in the case of negotiations. These results can be easily derived from a construction described in Theorem 1 of \cite{negI}. Given a deterministic linearly bounded automaton $A$ and a word $w$, Theorem 1 shows how to construct in polynomial time a nondeterministic negotiation $\Nn_A$ having one single run that simulates the execution of $A$ on the input $w$. In particular, we can set up $\Nn_A$ so that it is sound and its unique successful run executes a certain location, say $\ell$, if{}f $A$ accepts the input. Consider now a framework with $D=\{\bot, \top\}$, $\sem{\ell}$ as the function that always returns $\top$, and $\sem{\ell'}$ as the identity for every $\ell' \neq \ell$.
% the framework is Mazurkiewicz invariant, and deciding if the MOP is equal to $\top$ is
% polynomial for workflow graphs. However, it is  PSPACE-complete for sound negotiations. 
% The other two results are proved similarly. \javier{CHECK THIS!!}

\subsection{Mazurkiewicz-invariant analysis frameworks}
\label{subsec:mazur}

We introduce the notion of Mazurkiewicz equivalence between runs (also
called trace equivalence in the literature~\cite{book-of-traces}). 
Two equivalent runs started in the same
configuration will end up in the same configuration. 
We call an analysis framework Mazurkiewicz-invariant if the values of
equivalent runs are the same. 
We then show that the MOP of a Mazurkiewicz-invariant analysis is independent of the
scheduler.

\begin{definition}
\label{def:mazur}
Two nodes $n$, $m$ of a negotiation diagram are \emph{independent} if $\dom(n) \cap \dom(m) = \emptyset$.
Two \emph{locations are independent} if their nodes are independent.
Given two finite sequences of locations $w_1, w_2$, we write $w_1 \sim w_2$ if 
$w_1 = w \ell_1 \ell_2 w'$ and $w_2 = w \ell_2 \ell_1 w'$ for independent 
locations $\ell_1, \ell_2$. \emph{Mazurkiewicz equivalence}, denoted
by $\equiv$, is the reflexive-transitive closure of $\sim$. 
%A \emph{trace} is an equivalence class of $\equiv$. The trace of a
%sequence $w$ is denoted $[w]$. Given two traces $t, t'$, we write $t
%\sqsubseteq t'$ if $t=[w]$ and $t' =[wu]$ for some sequences $w, u$ of
%locations.
\end{definition}

The next lemma says that Mazurkiewicz equivalent runs have the same behaviors.

\begin{restatable}{lemma}{lemmaMazur}
\label{lem:mazur}
If $C_1\act{w}C_2$ and $v \equiv w$, then $C_1\act{v}C_2$. In
particular, if $w$ is a (successful) run, then $v$ is.
\end{restatable}

% The following lemma shows that, intuitively, the set of executable traces of a \emph{deterministic} negotiation is independent of the scheduler. In other words, the schedulers of deterministic negotiations do not influence which traces occur, but only which interleaving of each trace occurs. 

Interestingly Mazurkiewicz equivalence behaves very well with respect
to schedulers.

\begin{restatable}{lemma}{lemmaMazurii}
\label{lem:mazur2}
Let $\Nn$ be a deterministic negotiation diagram and let $S$ be a scheduler of
$\Nn$. For every successful run $w$ there is exactly one successful
run $v\equiv w$ that is compatible with $S$.
\end{restatable}

We observe that Lemma~\ref{lem:mazur2} may not hold for runs that are
not successful nor for non-deterministic negotiation  diagrams.
%see Appendix for details. 

%\begin{example}
%Lemma \ref{lem:mazur2} does not hold general negotiations. Consider the negotiation
%with ${\it Proc} = \{p, q\}$, $R= \{a\}$, $N = \{n_0, n_1, n_2, n_3, n_4\}$,
%$\dom(n_0)=\dom(n_2) = \dom(n_4)and
%$\delta(n_0, a, p) = \{n_1, n_2\}$, $\delta(n_0, a, q) = \{n_2, n_3\}$, and 
%$\delta(m,a,r) = \{n_4\}$ for every 

We can now define Mazurkiewicz invariant analysis frameworks, and
prove that they are independent of schedulers.

\begin{definition}
An analysis framework is \emph{Mazurkiewicz invariant} if 
$\sem{\ell_1} \circ \sem{\ell_2} = \sem{\ell_2} \circ \sem{\ell_1}$
for every two independent outcomes $\ell_1$, $\ell_2$.
\end{definition}

\begin{restatable}{theorem}{thmMazur}
\label{thm:mazur}
Let $\Nn$ be a negotiation diagram, and let $\sem{\_}$ be an analysis framework for 
$\Nn$. If $\Nn$ is deterministic and $\sem{\_}$ is Mazurkiewicz invariant, then 
$\sem{\Nn,S}=\sem{\Nn,S'}$ for every two schedulers $S, S'$, and so 
$\sem{\Nn} = \sem{\Nn,S}$ for every scheduler $S$.
\end{restatable}

It turns out that many interesting analysis frameworks are Mazurkiewicz
invariant, or Mazurkiewicz invariant under natural conditions. 
Let us look at the examples from Section~\ref{sec:examples}.

\medskip

\noindent {\bf The input/output framework} is Mazurkiewicz invariant if 
$\sem{\ell_1}(\sem{\ell_2}(\{v\})) = \sem{\ell_2}(\sem{\ell_1}(\{v\}))$ holds
whenever $\ell_1$ and $\ell_2$ are independent. This is not always the case, but holds when all
variables of $V$ are local variables. Formally, the set $V$ of variables is partitioned into sets $V_p$ of local variables for each process $p$. Further, 
$\sem{\ell}$ involves only the local variables of the processes involved in $\ell$: 
Letting $(v_\ell, v)$ denote a valuation of $V$, split into a valuation $v_\ell$ of the variables of the processes of $\ell$ and a valuation $v$ of  the rest, we have $\sem{\ell}(v_\ell, v) = (v'_\ell, v)$,
and $\sem{\ell}(v_\ell, v) = \sem{\ell}(v_\ell, v')$ for every $v_\ell, v, v'$.
\medskip

\noindent {\bf The anti-pattern framework} is not Mazurkiewicz
invariant. For example if we take some $\ell\in K$ independent of
$\ell_1$ then $\sem{\ell_1\ell\ell_2}\not=\sem{\ell\ell_1\ell_2}$. However, in 
Section \ref{sec:genkill} we will show that there is a Mazurkiewicz-invariant
framework for anti-patterns.\medskip

\noindent {\bf The minimal/maximal expected cost framework} is Mazurkiewicz invariant.
Indeed, it satisfies
$\sem{\ell_1} \circ \sem{\ell_2} =  \sem{\ell_2} \circ \sem{\ell_1}$ for all outcomes $\ell_1, \ell_2$,
independent or not. Further, by Theorem \ref{thm:mazur} the expected
cost is the same for every scheduler, and so the result of the
analysis is \emph{the} expected cost of the negotiation diagram.

\medskip

\noindent {\bf The best/worse case execution framework} is Mazurkiewicz invariant. 
Intuitively, the scheduler introduces an artificial linearization of
the nodes, which are 
however being executed in parallel. As in the previous case, the result of the analysis is the 
best-case/worst-case execution time of the negotiation diagram (if the
negotiation diagram is cyclic and the cycle non-zero time, then the worst-case execution time is infinite). 

\medskip

Our next goal is a generic algorithm for computing the MOP of
Mazurkiewicz-invariant frameworks for sound deterministic
negotiation diagrams. This will be done in Section~\ref{sec:compMOP}, but
before we will need some results on decomposing negotiation diagrams.

%%% Local Variables:
%%% mode: latex
%%% TeX-master: "../mlics"
%%% End:

\section{Decomposing Sound Negotiation Diagrams}
\label{sec:decomp}

We associate with every node $n$ and every location $\ell$ of a sound
deterministic negotiation diagram $\Nn$ a ``subnegotiation''
$\Nnn$ and $\Nnl$, and prove that it is also sound. In Section \ref{sec:compMOP}
we use these subnegotiations to define an
analysis algorithm sound deterministic negotiation diagrams. We
illustrate the results of this section on the example of
Figure \ref{fig:decomp}. 

\begin{figure}[ht]
\centerline{\scalebox{0.70}{\def\hs{2.5}
\def\vs{2.0}
\begin{tikzpicture}
\vnego[ports=3,id=n0,spacing=\vs]{0,0}
\node[above left = 1.2cm and -0.1 cm of n0_P0, font=\large] {$n_0$};
\node[left = 0.5cm of n0_P2, font=\large] {$p_1$};
\node[left = 0.5cm of n0_P1, font=\large] {$p_2$};
\node[left = 0.5cm of n0_P0, font=\large] {$p_3$};
\vnego[ports=1,id=n1]{\hs,2*\vs}
\node[above = 0.2cm  of n1_P0, font=\large] {$n_1$};
\vnego[ports=2,id=n2,spacing=\vs]{\hs,0}
\node[above left = 0.5cm and -0.1cm of n2_P0, font=\large] {$n_2$};
\vnego[ports=1,id=n3]{2*\hs,\vs}
\node[below left = 0cm and -0.1cm of n3_P0, font=\large] {$n_3$};
\vnego[ports=1,id=n4]{2*\hs,0.0}
\node[above left = 0cm and -0.1cm of n4_P0, font=\large] {$n_4$};
\vnego[ports=1,id=n5]{2*\hs,1.5*\vs}
\node[above left = -0.1cm and -0.1cm of n5_P0, font=\large] {$n_5$};
\vnego[ports=1,id=n6]{2*\hs,-0.5*\vs}
\node[below left = -0.1cm and -0.1cm of n6_P0, font=\large] {$n_6$};
\vnego[ports=2,id=n7,spacing=\vs]{3*\hs,0}
\node[above left = 0.5cm and -0.1cm of n7_P0, font=\large] {$n_7$};
\vnego[ports=3,id=n8,spacing=\vs]{4*\hs,0}
\node[above right = 1.0cm and -0.1 cm of n8_P0, font=\large] {$n_8$};

\pgfsetarrowsend{latex}
\draw (n0_P2) -- (n1_P0) node [above,midway]{$a$};
\draw (n0_P1) -- (n2_P1) node [above,midway]{$a$};
\draw (n0_P0) -- (n2_P0) node [above,midway]{$a$};
\draw (n1_P0) -- (n8_P2) node [above,midway]{$a$};
\draw (n2_P1) -- (n3_P0) node [above,midway]{$a$};
\draw (n2_P0) -- (n4_P0) node [above,midway]{$a$};
\draw (n3_P0) -- (n7_P1) node [above,midway]{$a$};
\draw (n4_P0) -- (n7_P0) node [above,midway]{$a$};
\draw (n3_P0) to [bend right = 40] node [right,midway]{$b$} (n5_P0);
\draw (n4_P0) to [bend left = 40] node [right,midway]{$b$} (n6_P0);
\draw (n5_P0) to [bend right = 40] node [left,midway]{$a$} (n3_P0);
\draw (n6_P0) to [bend left = 40] node [left,midway]{$a$} (n4_P0);
\draw (n7_P1) -- (n8_P1) node [above,midway]{$a$};
\draw (n7_P0) -- (n8_P0) node [above,midway]{$a$};
\draw (n7_P1) to [bend right = 85] node [above,near start]{$b$} (n2_P1);
\draw (n7_P0) to [bend left = 85] node [above,near start]{$b$} (n2_P0);

\end{tikzpicture}}}
\caption{\small A negotiation diagram with three processes.}
\label{fig:decomp}
\end{figure}
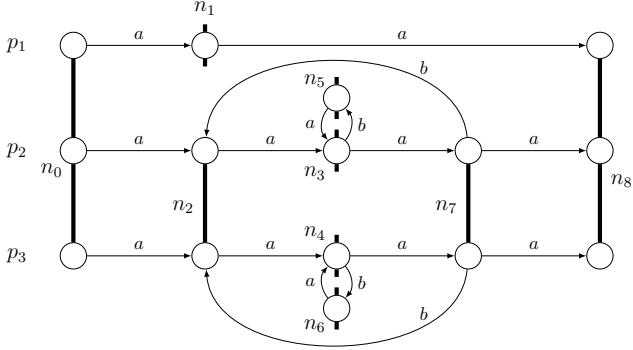

Intuitively, $\Nnn$ contains the nodes $n'$ such that $\dom(n') \subseteq \dom(n)$, 
with transitions inherited from $\N$, and $n$ as initial node. The non-trivial part is to define the final node
and show that $\Nnn$ is sound. Given a location $\ell = (n, a)$, the
negotiation $\Nnl$ contains the part of $\Nnn$ reachable by executions that start with $\ell$
and afterwards only use nodes with domains \emph{strictly} included in
$\dom(n)$. Figure \ref{fig:subnegn3} shows some of the subnegotiations we will obtain for 
some nodes and locations of Figure \ref{fig:decomp}. 

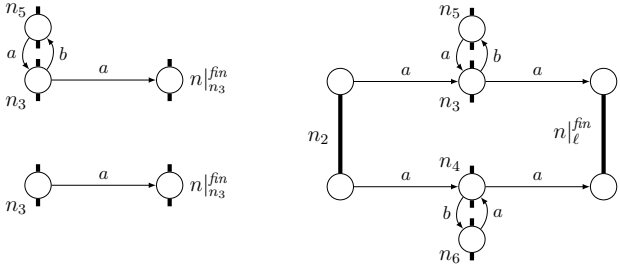
\begin{figure}[ht]
\raisebox{0.65cm}{\scalebox{0.70}{\def\hs{2.5}
\def\vs{2.0}
\begin{tikzpicture}
\begin{scope}
%subnego N_n3
%\vnego[ports=3,id=n0,spacing=\vs]{0,0}
%\node[above left = 1.2cm and -0.1 cm of n0_P0, font=\large] {$n_0$};
%\node[left = 0.5cm of n0_P2, font=\large] {$p_1$};
%\node[left = 0.5cm of n0_P1, font=\large] {$p_2$};
%\node[left = 0.5cm of n0_P0, font=\large] {$p_3$};
%\vnego[ports=1,id=n1]{\hs,2*\vs}
%\node[above = 0.2cm  of n1_P0, font=\large] {$n_1$};
%\vnego[ports=2,id=n2,spacing=\vs]{\hs,0}
%\node[above left = 0.5cm and -0.1cm of n2_P0, font=\large] {$n_2$};
\vnego[ports=1,id=n3]{2*\hs,\vs}
\node[below left = 0cm and -0.1cm of n3_P0, font=\large] {$n_3$};
%\vnego[ports=1,id=n4]{2*\hs,0.0}
%\node[above left = 0cm and -0.1cm of n4_P0, font=\large] {$n_4$};
\vnego[ports=1,id=n5]{2*\hs,1.5*\vs}
\node[above left = -0.1cm and -0.1cm of n5_P0, font=\large] {$n_5$};
%\vnego[ports=1,id=n6]{2*\hs,-0.5*\vs}
%\node[below left = -0.1cm and -0.1cm of n6_P0, font=\large] {$n_6$};
%\vnego[ports=2,id=n7,spacing=\vs]{3*\hs,0}
%\node[above left = 0.5cm and -0.1cm of n7_P0, font=\large] {$n_7$};
%\vnego[ports=3,id=n8,spacing=\vs]{4*\hs,0}
%\node[above right = 1.0cm and -0.1 cm of n8_P0, font=\large] {$n_8$};
\vnego[ports=1,id=nf,spacing=\vs]{3*\hs,\vs}
\node[right = 0.0cm of nf_P0, font=\large] {$n|^\fin_{n_3}$};

\pgfsetarrowsend{latex}
%\draw (n0_P2) -- (n1_P0) node [above,midway]{$a$};
%\draw (n0_P1) -- (n2_P1) node [above,midway]{$a$};
%\draw (n0_P0) -- (n2_P0) node [above,midway]{$a$};
%\draw (n1_P0) -- (n8_P2) node [above,midway]{$a$};
%\draw (n2_P1) -- (n3_P0) node [above,midway]{$a$};
%\draw (n2_P0) -- (n4_P0) node [above,midway]{$a$};
%\draw (n3_P0) -- (n7_P1) node [above,midway]{$a$};
\draw (n3_P0) -- (nf_P0) node [above,midway]{$a$};
%\draw (n4_P0) -- (n7_P0) node [above,midway]{$a$};
\draw (n3_P0) to [bend right = 40] node [right,midway]{$b$} (n5_P0);
%\draw (n4_P0) to [bend right = 40] node [left,midway]{$b$} (n6_P0);
\draw (n5_P0) to [bend right = 40] node [left,midway]{$a$} (n3_P0);
%\draw (n6_P0) to [bend right = 40] node [right,midway]{$a$} (n4_P0);
%\draw (n7_P1) -- (n8_P1) node [above,midway]{$a$};
%\draw (n7_P0) -- (n8_P0) node [above,midway]{$a$};
%\draw (n7_P1) to [bend right = 85] node [above,near start]{$b$} (n2_P1);
%\draw (n7_P0) to [bend left = 85] node [above,near start]{$b$} (n2_P0);
\end{scope}

\begin{scope}[shift={(0,-2)}]
%subnego N_n3
%\vnego[ports=3,id=n0,spacing=\vs]{0,0}
%\node[above left = 1.2cm and -0.1 cm of n0_P0, font=\large] {$n_0$};
%\node[left = 0.5cm of n0_P2, font=\large] {$p_1$};
%\node[left = 0.5cm of n0_P1, font=\large] {$p_2$};
%\node[left = 0.5cm of n0_P0, font=\large] {$p_3$};
%\vnego[ports=1,id=n1]{\hs,2*\vs}
%\node[above = 0.2cm  of n1_P0, font=\large] {$n_1$};
%\vnego[ports=2,id=n2,spacing=\vs]{\hs,0}
%\node[above left = 0.5cm and -0.1cm of n2_P0, font=\large] {$n_2$};
\vnego[ports=1,id=n3]{2*\hs,\vs}
\node[below left = 0cm and -0.1cm of n3_P0, font=\large] {$n_3$};
%\vnego[ports=1,id=n4]{2*\hs,0.0}
%\node[above left = 0cm and -0.1cm of n4_P0, font=\large] {$n_4$};
%\vnego[ports=1,id=n5]{2*\hs,1.5*\vs}
%\node[above left = -0.1cm and -0.1cm of n5_P0, font=\large] {$n_5$};
%\vnego[ports=1,id=n6]{2*\hs,-0.5*\vs}
%\node[below left = -0.1cm and -0.1cm of n6_P0, font=\large] {$n_6$};
%\vnego[ports=2,id=n7,spacing=\vs]{3*\hs,0}
%\node[above left = 0.5cm and -0.1cm of n7_P0, font=\large] {$n_7$};
%\vnego[ports=3,id=n8,spacing=\vs]{4*\hs,0}
%\node[above right = 1.0cm and -0.1 cm of n8_P0, font=\large] {$n_8$};
\vnego[ports=1,id=nf,spacing=\vs]{3*\hs,\vs}
\node[right = 0.0cm of nf_P0, font=\large] {$n|^\fin_{n_3}$};

\pgfsetarrowsend{latex}
%\draw (n0_P2) -- (n1_P0) node [above,midway]{$a$};
%\draw (n0_P1) -- (n2_P1) node [above,midway]{$a$};
%\draw (n0_P0) -- (n2_P0) node [above,midway]{$a$};
%\draw (n1_P0) -- (n8_P2) node [above,midway]{$a$};
%\draw (n2_P1) -- (n3_P0) node [above,midway]{$a$};
%\draw (n2_P0) -- (n4_P0) node [above,midway]{$a$};
%\draw (n3_P0) -- (n7_P1) node [above,midway]{$a$};
\draw (n3_P0) -- (nf_P0) node [above,midway]{$a$};
%\draw (n4_P0) -- (n7_P0) node [above,midway]{$a$};
%\draw (n3_P0) to [bend right = 40] node [right,midway]{$b$} (n5_P0);
%\draw (n4_P0) to [bend right = 40] node [left,midway]{$b$} (n6_P0);
%\draw (n5_P0) to [bend right = 40] node [left,midway]{$a$} (n3_P0);
%\draw (n6_P0) to [bend right = 40] node [right,midway]{$a$} (n4_P0);
%\draw (n7_P1) -- (n8_P1) node [above,midway]{$a$};
%\draw (n7_P0) -- (n8_P0) node [above,midway]{$a$};
%\draw (n7_P1) to [bend right = 85] node [above,near start]{$b$} (n2_P1);
%\draw (n7_P0) to [bend left = 85] node [above,near start]{$b$} (n2_P0);
\end{scope}

\end{tikzpicture}}} \hspace{0.5cm}\scalebox{0.70}{\def\hs{2.5}
\def\vs{2.0}
\begin{tikzpicture}
%subnego N_n3
%\vnego[ports=3,id=n0,spacing=\vs]{0,0}
%\node[above left = 1.2cm and -0.1 cm of n0_P0, font=\large] {$n_0$};
%\node[left = 0.5cm of n0_P2, font=\large] {$p_1$};
%\node[left = 0.5cm of n0_P1, font=\large] {$p_2$};
%\node[left = 0.5cm of n0_P0, font=\large] {$p_3$};
%\vnego[ports=1,id=n1]{\hs,2*\vs}
%\node[above = 0.2cm  of n1_P0, font=\large] {$n_1$};
\vnego[ports=2,id=n2,spacing=\vs]{\hs,0}
\node[above left = 0.5cm and -0.1cm of n2_P0, font=\large] {$n_2$};
\vnego[ports=1,id=n3]{2*\hs,\vs}
\node[below left = 0cm and -0.1cm of n3_P0, font=\large] {$n_3$};
\vnego[ports=1,id=n4]{2*\hs,0.0}
\node[above left = 0cm and -0.1cm of n4_P0, font=\large] {$n_4$};
\vnego[ports=1,id=n5]{2*\hs,1.5*\vs}
\node[above left = -0.1cm and -0.1cm of n5_P0, font=\large] {$n_5$};
\vnego[ports=1,id=n6]{2*\hs,-0.5*\vs}
\node[below left = -0.1cm and -0.1cm of n6_P0, font=\large] {$n_6$};
\vnego[ports=2,id=n7,spacing=\vs]{3*\hs,0}
\node[above left = 0.5cm and -0.1cm of n7_P0, font=\large] {$n|^\fin_\ell$};
%\vnego[ports=3,id=n8,spacing=\vs]{4*\hs,0}
%\node[above right = 1.0cm and -0.1 cm of n8_P0, font=\large] {$n_8$};

\pgfsetarrowsend{latex}
%\draw (n0_P2) -- (n1_P0) node [above,midway]{$a$};
%\draw (n0_P1) -- (n2_P1) node [above,midway]{$a$};
%\draw (n0_P0) -- (n2_P0) node [above,midway]{$a$};
%\draw (n1_P0) -- (n8_P2) node [above,midway]{$a$};
\draw (n2_P1) -- (n3_P0) node [above,midway]{$a$};
\draw (n2_P0) -- (n4_P0) node [above,midway]{$a$};
\draw (n3_P0) -- (n7_P1) node [above,midway]{$a$};
\draw (n4_P0) -- (n7_P0) node [above,midway]{$a$};
\draw (n3_P0) to [bend right = 40] node [right,midway]{$b$} (n5_P0);
\draw (n4_P0) to [bend right = 40] node [left,midway]{$b$} (n6_P0);
\draw (n5_P0) to [bend right = 40] node [left,midway]{$a$} (n3_P0);
\draw (n6_P0) to [bend right = 40] node [right,midway]{$a$} (n4_P0);
%\draw (n7_P1) -- (n8_P1) node [above,midway]{$a$};
%\draw (n7_P0) -- (n8_P0) node [above,midway]{$a$};
%\draw (n7_P1) to [bend right = 85] node [above,near start]{$b$} (n2_P1);
%\draw (n7_P0) to [bend left = 85] node [above,near start]{$b$} (n2_P0);
\end{tikzpicture}}
\caption{\small Subnegotiations $\Nn|_{n_3}$ (top left), $\Nn|_{(n_3,a)}$ (bottom left),  and $\Nn|_{(n_2,a)}$ (right) 
of the negotiation diagram of Figure \ref{fig:decomp}.
Nodes unreachable from the initial node are not shown.}
\label{fig:subnegn3}
\end{figure}

The rest of the section first presents a theorem 
showing the existence and uniqueness of some special configurations (Section~\ref{subsec:unique}), and then uses
it to define $\Nnn$ and $\Nnl$, and prove their soundness (Section \ref{subsec:subneg}).

\subsection{Unique maximal configurations}
\label{subsec:unique}

Given a node $m$ of a sound and deterministic negotiation, we prove the existence 
of a unique reachable configuration $I(m)$ enabling $m$ and only $m$. Then we show the following:
if we start from $I(m)$ as initial configuration, ``freeze'' the processes 
of $\Proc \setminus \dom(n)$, and let the processes of $\dom(n)$ execute maximally 
(i.e., until they cannot execute any node without the help of  
processes of $\Proc \setminus \dom(n)$), then we \emph{always} reach the same ``final''
configuration $F(m)$.  Additionally, given a location $\ell=(m, a)$ we show:
If from $I(m)$ we execute $\ell$ and let the processes of $\dom(n)$ 
proceed until no enabled node $n$ satisfies $\dom(n) \subset \dom(m)$, 
then again we always reach the same ``final'' configuration $F(\ell)$.

Let $X \subseteq \Proc$ be a set of processes. A sequence of locations $\ell_1,\dots,\ell_k$ is an \emph{$X$-sequence}
if the domains of all $\ell_i$ are included in $X$; it is a \emph{strict $X$-sequence} if 
moreover the domains of all $\ell_i$, but possibly $\ell_1$, are strictly included in $X$. We write 
(strict) $n$-sequence for  (strict) $\dom(n)$-sequence.

We write $n' \domleq n$ if $\dom(n') \subseteq
\dom(n)$, and $n' \domle n$ if $\dom(n') \subsetneq
\dom(n)$ (and similarly for $\ell' \domleq \ell, \ell' \domle
\ell$, $\ell \domleq n$ etc). 
Our goal is to prove:

\begin{restatable}{theorem}{thmUnique}
\label{thm:unique}
Let $m$ be a reachable node of a sound deterministic negotiation
diagram $\Nn$.
\begin{itemize}
\item[(i)] There is a unique reachable configuration $\first{m}$ of $\Nn$ that enables $m$, and no other node. 
\item[(ii)] There is a unique configuration $\last{m}$ such that
\begin{itemize}
\item $\last{m}$ is reachable from $\first{m}$ by means of an $m$-sequence, and
\item for every node $n$ enabled at $\last{m}$, $\dom(n)$ is not included in $\dom(m)$.
%, i.e., $n$ needs a process outside of $\dom(m)$.
\end{itemize}
\item[(iii)] For every location $\ell$ of $m$ there is a unique configuration $\last{\ell}$ such that
\begin{itemize}
\item $\last{\ell}$ is reachable from $\first{m}$ by means of a strict
  $m$-sequence starting with $\ell$, and 
\item for every node $n$ enabled at $\last{\ell}$, $\dom(n)$ is not strictly included in $\dom(m)$.
%, i.e., needs all processes of $\dom(m)$ or a process outside of $\dom(m)$.
\end{itemize}
\end{itemize} 
\end{restatable}
\noindent E.g., in Figure \ref{fig:decomp} we have $I(n_1) = (n_1,n_8,n_8)$ (an abbreviation
for $I(n_1)(p_1) = \{n_1\},
I(n_1)(p_2) = \{n_8\}, I(n_1)(p_3) = \{n_8\}$); $I(n_2) = (n_8,n_2,n_2)$; and $I(n_3) = (n_8, n_3, n_7)$.
Moreover, $F(n_1) = F(n_2) = (n_8, n_8, n_8)$; and $F(n_3) = (n_8, n_7, n_7)$. Further, we get
$F(n_7,a)=(n_8,n_7,n_7)$ and $F(n_8,b)=(n_8,n_2,n_2)$.

The proof of Theorem \ref{thm:unique} is quite involved. The theorem is a
consequence of the Unique Configuration lemma (Lemma~\ref{lemma:main-technical} below),
which relies on the Domination lemma (Lemma~\ref{lem:domin} below), 
which in turn is based on results of \cite{DBLP:conf/concur/EsparzaKMW16,negII}.

%\subsubsection{Domination lemma}
%\label{subsec:domin}

 A~\emph{local path} of a negotiation diagram $\Nn$ is a path
$n_0\act{p_0,a_0}n_1\act{p_1,a_1}\dots\act{p_{k-1},a_{k-1}} n_k$ in
the graph of $\Nn$. A local path is a \emph{local circuit} if  $k>0$
and $n_0 = n_k$. 
A local path is \emph{reachable} if some node in the path is reachable.
The domination lemma says that every local circuit has a dominant action.

\begin{restatable}{lemma}{lemmaDomin}(\textbf{Domination Lemma})\ 
\label{lem:domin}
Let $\Nn$ be a deterministic sound negotiation diagram.
Every reachable local circuit of $\Nn$ contains a
dominant node, i.e. a node $n$ such that $m\domleq n$,
for every node $m$ of the circuit. 
\end{restatable}

The unique configuration lemma says that if two enabled configurations
agree on a set of processes $X$ and every enabled action in one of the
two configurations needs a process from $X$, then the two
configurations are actually the same.
\begin{restatable}{lemma}{lemmaMainTechnical}(\textbf{Unique Configuration Lemma})\  
\label{lemma:main-technical}
  Let $X\incl \Proc$ be a set of processes. Let $C_1,C_2$ be reachable
configurations such that (1) $C_1(p)=C_2(p)$ for every $p \in X$, and
(2) every node $n$ enabled at $C_1$ or $C_2$ satisfies $\dom(n) \cap X \not=\es$. Then $C_1=C_2$.
\end{restatable}
The above lemma gives Theorem~\ref{thm:unique} rather directly. 
For (i), take $X=\dom(m)$. Suppose 
that there are two configurations $I_1$ and $I_2$ as in (i). The hypotheses of Lemma~\ref{lemma:main-technical} are
satisfied, and so $I_1=I_2$. The case (ii) is equally easy, while
(iii) is only a bit more involved.

%%% Local Variables:
%%% mode: latex
%%% TeX-master: "mlics"
%%% End:

\subsection{Subnegotiations for nodes and locations}
\label{subsec:subneg}

We use Theorem \ref{thm:unique} to define the subnegotiations
$\Nn|_n$ and $\Nn|_\ell$ for each node $n$ and location $\ell$ of a sound deterministic 
negotiation diagram $\Nn$, and prove that they are sound.

\begin{definition}\label{def:n-restriction}
Let $\Nn$ be a sound deterministic negotiation diagram and let $n$ be a reachable node
of $\Nn$. The negotiation diagram $\Nn|_n$ contains all the nodes and locations that 
appear in the $n$-sequences $u$ such that $\first{n} \act{u} \last{n}$, plus a new
final node $n|^\fin_n$ with $\dom(n|^\fin_n)=\dom(n)$. The initial node is $n$, 
and the transition function $\d|_n$ is defined as follows. For given $m,a,p$, we set: 
\begin{equation*}
\d|_n(m,a,p)=
  \begin{cases}
    \d(m,a,p) & \text{ if $\d(m,a,p) \neq \last{n}(p)$ } \\
    n|^\fin_n & \text{ if $\d(m,a,p) = \last{n}(p)$}
  \end{cases}
\end{equation*}
\end{definition}

An example of $\Nnn$ is given on the left of Figure~\ref{fig:subnegn3}.
% shows $\Nn|_{n3}$ of the negotiation diagram
%of Figure \ref{fig:decomp}.

\begin{restatable}{lemma}{lemmansound}
\label{lem:nsound}
 If $\Nn$ is a sound deterministic negotiation diagram then so is $\Nn|_n$.
\end{restatable}

The subnegotiation $\Nn|_\ell$ induced by a location $\ell=(n,a)$ is defined analogously to $\N|_n$, 
with two differences. First, in $\N|_\ell$ the node $n$ has $a$ as the unique outcome. 
Second, the domain of every node of  $\Nn|_\ell$, except the node $n$ itself, 
is \emph{strictly} included in $\dom(n)$. 

\begin{definition}\label{def:l-restriction}
Let $\Nn$ be a deterministic sound negotiation diagram, let $n$ be a reachable node
of $\Nn$, and let $\ell=(n,a)$ for some outcme $a$ of $n$. The negotiation diagram $\Nn|_\ell$ contains all 
the nodes and locations that appear in the strict $n$-sequences $u$
such that  
$\first{n} \act{\ell \, u} \last{\ell}$, plus a new
final node $n|^\fin_\ell$. The initial node is $n$, it has the unique
outcome $a$, and the transition function 
$\d|_\ell$ is defined as follows. For given $m,b,p$  we set: 
\begin{equation*}
\d|_\ell(m,b,p)=
  \begin{cases}
    \d(m,b,p) & \text{ if $\d(m,b,p) \neq \last{\ell}(p)$ } \\
    n|^\fin_\ell & \text{otherwise}
  \end{cases}
\end{equation*}
\end{definition}
Figure \ref{fig:subnegn3} shows (on the right) $\Nn|_{\ell}$ for the location $\ell=(n_2, a)$ of the negotiation diagram of Figure \ref{fig:decomp}.

%\begin{figure}[ht]
%\centerline{\scalebox{0.70}{\input{fig-subnegn2a}}}
%\caption{\small Subnegotiation $\Nn|_\ell$ of the negotiation diagram of Figure \ref{fig:decomp}
%for $\ell=(n_2,a)$.}
%\label{fig:subnegn2a}
%\end{figure}

\begin{restatable}{lemma}{lemlsound}
\label{lem:lsound}
 If $\Nn$ is a sound deterministic negotiation diagram, then so is $\Nn|_\ell$.
\end{restatable}

%%% Local Variables:
%%% mode: latex
%%% TeX-master: "../mlics"
%%% End:

\section{Computing the MOP}
\label{sec:compMOP}

We use the decomposition of Section \ref{sec:decomp} to define an
algorithm computing the MOP for Mazurkiewicz-invariant frameworks and
arbitrary sound deterministic 
negotiation diagram. 
The idea is to repeatedly reduce parts of the negotiation diagram without
changing the meaning of the whole negotiation.
When reduction will be no longer possible, the negotiation diagram will have
only one location, whose value will be the value of the negotiation. 
The goal of this section is to present Algorithm~\ref{alg:red} and
Theorem~\ref{thm:main} that is our main result.

As we have seen in the previous section, for every node $n$,
the negotiation diagram $\Nnn$ is sound and the final configuration $F(n)$ is
unique. 
Thus we can safely replace all the transitions from $n$ by one transition
going directly to $F(n)$ and assign to this transition the value of
$\Nnn$. 
This requires to be able to compute $F(n)$ as well as the value of $\Nnn$.
For this we proceed by induction on the domain of $n$ starting from
nodes with the smallest domain. 
As we will see, this will require us to compute MOP only for
negotiations of two (very) special forms
% The algorithm computes the MOP for increasingly larger
% subnegotiations, until it produces the final result. It calls a
% procedure ${\it MOP}(\Nn)$ that computes the MOP of $\Nn$ for two
% (very) special classes of sound deterministic negotiations:

\medskip

\noindent \textbf{One-trace negotiations}. These are acyclic
negotiation diagrams in which every node has one single outcome. In
this degenerate case, all the executions of the negotiation diagram
are Mazurkiewicz equivalent; moreover, by acyclicity, the trace contains every location at most once. Since the analysis framework is Mazurkiewicz-invariant, we have $\sem{\Nn} = \sem{w}$ for any 
successful run $w$. A successful run can be computed by just executing
the negotiation diagram with some arbitrary scheduler. Once  a successful run
$w$ is computed, we extract from it a flow-graph with $|w|$ nodes,
that is actually a sequence, and compute MFP. 

\medskip

\noindent \textbf{Replications}. Intuitively, a replication is a
negotiation diagram in which all processes are involved in every node,
and all processes move uniformly, that is, after they agree on an
outcome they all move to the same node. Formally, a negotiation
diagram is a \emph{replication} if 
for every reachable node $n$ and every outcome $(n,a)$ there is a node
$m$ such that $\delta(n,a, p)=m$ for every process $p$. Observe that,
in particular, all nodes of a replication have the same domain. It
follows that (the reachable part of) a replication is a flow-graph
``in disguise''. More precisely, we can assign to it a flow-graph
having one node for every reachable node, and an edge for every
location $(n,a)$, leading from $n$ to $\delta(n,a, p)$, where $p$ can
be chosen arbitrarily out of $\dom(n)$. It follows immediately that
the MOP for the negotiation diagram is equal to the MFP of this flow-graph.

\medskip

% The algorithm, presented as Algorithm \ref{alg:red}, proceeds by repeatedly computing subnegotiations $\Nn|_\ell$ that are one-trace negotiations, or subnegotiations $\Nn|_n$ that are replications. In both cases, the algorithm calls the {\it MOP} routine, and, loosely speaking, replaces the complete subnegotiation by an equivalent one with a single outcome.  We formaly define the replacements, and prove that they
% maintain the MOP.

\begin{definition}
A node $n$ is \emph{reduced} if it has one single outcome, and for
this single outcome $a$ we have $\delta(n,a,p)=F(n)(p)$ for every $p
\in \dom(n)$. 

A location $\ell=(n,a)$ is \emph{reduced} if $\d(n,a,p)
=F(\ell)(p)$ for every  $p \in \dom(n)$.
\end{definition}

Now we will define an operation of reducing nodes and locations in a
negotiation diagram; this is the core operation of Algorithm~\ref{alg:red}.
For a location $\ell=(n,a)$, the operation  $\Red_{\ell}(\Nn)$ removes
the transition of $n$ on $\ell$ and adds a new transition on $(n,a_\ell)$.
Similarly $\Red_n(\Nn)$, but this time it removes all 
transitions from $n$ and adds a single new transition on $(n,a_n)$.
\begin{definition}
\label{def:n_l}
Let $\Nn=\struct{\Proc,N,\dom,R,\d}$ be a sound deterministic
negotiation diagram and let $\ell=(n,a)$ be a non-reduced outcome of $\Nn$. 
The negotiation diagram $\Redl(\Nn)$ has the same components as
$\Nn$ but for
$\out$ and $\d$ that are subject to the following changes:
\begin{itemize}
\item $\out(n):=(out(n)\setminus\set{a})\cup\set{a_\ell}$;
\item $\delta(n,a_\ell, p) := F(\ell)(p)$ for every process $p \in \dom(n)$.
\end{itemize}
The negotiation diagram $\Redn(\Nn)$ is defined similarly but now:
\begin{itemize}
\item $\out(n):=\set{a_n}$;
\item $\delta(n,a_{n}, p) := F(n)(p)$ for every process $p \in \dom(n)$.
\end{itemize}
\end{definition}

The next lemma states that these reduction
operations preserve the meaning of a negotiation diagram. 

\begin{restatable}{lemma}{lemEquivl}
\label{lem:equivl}
Let $\Nn$ and $\ell=(n,a)$ be as in Definition \ref{def:n_l}. Assign
to the new location $\ell'=(n, a_\ell)$ the mapping $\sem{\ell'} :=
\sem{\Nnl}$. Then $\sem{\Nn}=\sem{{\it Red}_\ell(\Nn)}$.
Analogously, $\sem{\Nn}=\sem{\Redn(\Nn)}$ when we assign
$\sem{(n,a_n)}=\sem{\Nnn}$.
\end{restatable}

% \begin{definition}
% \label{def:n_n}
% Let $\Nn=\struct{\Proc,N,\dom,R,\d}$ be a sound deterministic negotiation and let $n \in N$ be a non-reduced node.
% % such that $\Nn|_n$ is a replication. 
% Let $N|_n$ be the set of nodes of $\Nn|_n$. The negotiation $\mathit{Red}_n(\Nn)=\struct{\Proc',N',\dom',R',\d'}$ is defined as follows:
% \begin{itemize}
% \item $\Proc'=\Proc$, $N'=N$, and $\dom'=\dom$.
% \item $R'$ is the disjoint union of $R$ and a set $\{ a_{n'} \mid n' \in N|_n\}$ of fresh results.
% \item $\out'(n')=\out(n')$ for every $n' \in N \setminus N|_n$, and 
% $\d'(n',a,p)=\d(n',a,p)$ for every result $a \in \out(n')$ and process $p$.
% \item $\out(n') = \{ a_{n'} \}$ for every $n' \in N$, and $\delta'(n',a_{n'}, p) = F(n')(p)$ for every process $p \in \dom(n')$.
% \end{itemize}
% \end{definition}

\begin{algorithm}[t]
\caption{\small Algorithm computing MOP for a sound deterministic negotiation
  diagram $\Nn$.}
\begin{algorithmic}[1]
\While{$\Nn$ has non-reduced nodes}
\State $m$ := $\prec$-minimal, non-reduced node
\State  $X=\dom(m)$
\For{every $\ell=(n,a)$ with $\dom(n)=X$}
\State \#\# $\Nn|_\ell$ is a one-trace negotiation \#\#
\State $\sem{n,a_\ell}:= {\it MOP}(\Nn|_\ell)$
\State $\Nn := \Redl(\Nn)$
\EndFor
\For{every node $n$ such that $\dom(n)=X$}
\State \#\# $\Nnn$ is a replication \#\#
\State $\sem{n,a_{n}} = {\it MOP}(\Nn|_{n})$
\EndFor
\For{every node $n$ such that $\dom(n)=X$}
\State $\Nn := \Redn(\Nn)$ 
\EndFor
\EndWhile
\State return $\sem{\ell}$, where $\ell$ is the unique outcome of the initial node of $\Nn$
\end{algorithmic}
\label{alg:red}
\end{algorithm}

At this point we can examine Algorithm~\ref{alg:red}. 
The algorithm repeatedly applies reduction operations to a given
negotiation diagram.
Thanks to Lemma~\ref{lem:equivl} these reductions preserve the meaning
of the negotiation diagram. 
At every reduction, the number of reachable locations in the
negotiation diagram decreases. 
So the algorithm stops, and when it stops the negotiation diagram has only one
reachable location. 
The abstract semantics of this location is equal to the abstract
semantics of the original negotiation diagram.

This argument works if indeed we can compute $MOP(\Nnl)$ and
$MOP(\Nnn)$ in lines $6$ and $11$ of the algorithm, respectively. 
For this it is enough to show that the invariants immediately
preceding these lines hold, as this would mean that we deal with
special cases we have discussed at the beginning of this section.
The following lemma implies that the invariants indeed hold.

\begin{restatable}{lemma}{invariants}
Let $\Nn$ be a sound deterministic negotiation diagram, and $n$ a node such
that all nodes $m\fle n$ are reduced in $\Nn$.
\begin{itemize}
 \item[(1)]  if $a$ is an outcome of $n$, then for $\ell=(n,a)$ the
   negotiation diagram $\Nnl$ is a one-trace negotiation.
 \item[(2)] if all locations $\ell'\fleq n$ are reduced, then $\Nnn$ is a replication.
   \end{itemize}
 \end{restatable}
 
\begin{example}
Consider the negotiation diagram of Figure \ref{fig:decomp}. 
Assume that all locations have cost 1, and that the probability of a location
$\ell=(n, a)$ is $1/|{\it out}(n)|$ (so, for example, the locations
$(n_3, a)$ and $(n_3, b)$ have probability 1/2, while $(n_0, a)$ has
probability 1). We compute the expected cost of the
diagram using Algorithm \ref{alg:red}. 

The minimal non-reduced nodes w.r.t. $\preceq$ are $n_3, n_4, n_5, n_6$. 
All their locations satisfy $\sem{\Nn|_\ell}=\sem{\ell}$. The algorithm computes 
${\it MOP}(\Nn|_{n_i})$ for $i=3,4,5,6$. The subnegotiations $\Nn|_{n_3}$ and $\Nn|_{n_5}$ are shown in Figure \ref{fig:subnegn3}; the other two are similar. ${\it MOP}(\Nn|_{n_3})$ is the expected cost of reaching $n_{fin}$ from $n_3$ in $\Nn|_{n_3}$. Since 
$\Nn|_{n_3}$ is a replication (in fact, it is even a flow-graph), we can compute it as the least solution of the following fixed point equation, where we abbreviate ${\it Prob}(\ell)$ to $P(\ell)$ and
${\it Cost}(\ell)$ to $C(\ell)$:
\begin{align*}
(p, c) = & \bigg(p \cdot P(n_3,b) \cdot P(n_5, a) + P(n_3, a), \\
& P(n_3, b) \cdot P(n_5, a) \cdot (C(n_3, b) + C(n_5,a) + c ) + \\
& P(n_3, a) \cdot C(n_3, a)\bigg)
\end{align*}
\noindent which gives
\begin{align*}
(p, c) = & \bigg(\frac{1}{2}p +\frac{1}{2},  \frac{1}{2}(2+c)+ \frac{1}{2} \bigg)
\end{align*}
\noindent with least fixed point $(1, 3)$, which of course can be computed by just solving the linear equation. So ${\it MOP}(\Nn|_{n_3})=(1,3)$. We obtain $\sem{n_3, a_{n_3}} = \sem{n_4, a_{n_4}}
= 3$, $\sem{n_5, a_{n_5}} = \sem{n_6, a_{n_6}} = 4$, and the
reduced negotiation diagram at the top of Figure \ref{fig:intermediate}. Observe that nodes $n_5$ and $n_6$ are no longer reachable. 

The minimal non-reduced nodes are now $n_2$ and $n_7$. The locations of $n_7$
satisfy $\sem{\Nn|_\ell}=\sem{\ell}$. For the location $(n_2,a)$ we obtain
${\it MOP}(\Nn|_{(n_2, a)})=(1, 1+3+3)=(1,7)$; this we can easily do
% 7 is the sum of the costs of (n2,a), (n3, a_n3), and (n4, a_n4), which are 
% equal to 1, 3, and 3, respectively.
because $\Nn|_{(n_2, a)}$ is a one-trace negotiation; we just pick any successful run,
for example $(n_2, a)(n_3, a_{n_3})(n_4, a_{n_4})$, and add the costs. After reducing
$\Nn|_{(n_2, a)}$ we obtain the negotiation diagram in the middle of
Figure \ref{fig:intermediate}; the non-reachable nodes $n_3, \ldots,
n_6$ are no longer displayed, and all locations have cost 1 but $(n_2, a_{(n_2, a)})$, which
has cost 7. Further, all locations $\ell$ of $n_2$ and $n_7$ satisfy
$\sem{\Nn|_\ell}=\sem{\ell}$. We compute $\Nn|_{n_2}$ and $\Nn|_{n_7}$
by a fixed point calculation analogous to the one above
(observe that both of them are replications), and after reduction
obtain the negotiation diagram at the bottom of the figure, with $\sem{n_7, a_{n_7}} = (1,9)$
%The expected cost from n7  is  the least fixed point of
%  c =   1/2  C(n7,a)   +   1/2 (C(n7,b)+C(n2,a_(n2,a)+c) 
%     =  1/2    1          +   1/2 (    1     +      7           +c)
%     =  1/2                +   4    +   1/2 c
% That is,    1/2 c =  1/2 + 4,  and so  c=9
and $\sem{n_2, a_{n_2}} = (1,7+9)=(1,16)$. 
% Because the cost from n2 is the expected cost of  N_{n2,a}, which is equal to 7,
% plus the expected cost from n7, which is equal to 9
Now, the minimal non-reduced node is $n_0$.We compute ${\it MOP}(\Nn|_{(n_0, a)})$ 
($\Nn|_{(n_0, a)}$ is a one-trace negotiation), and obtain 
$\sem{n_0, a}=(1, 1+1+16)=(1,18)$, which is the final result.
% The cost of N_(n0,a)  is the cost of (n0,a), which is equal to 1, plus the cost of (n1, a), which
% is equal to 1, plus the expected cost from n2, which is equal to 16
\end{example}

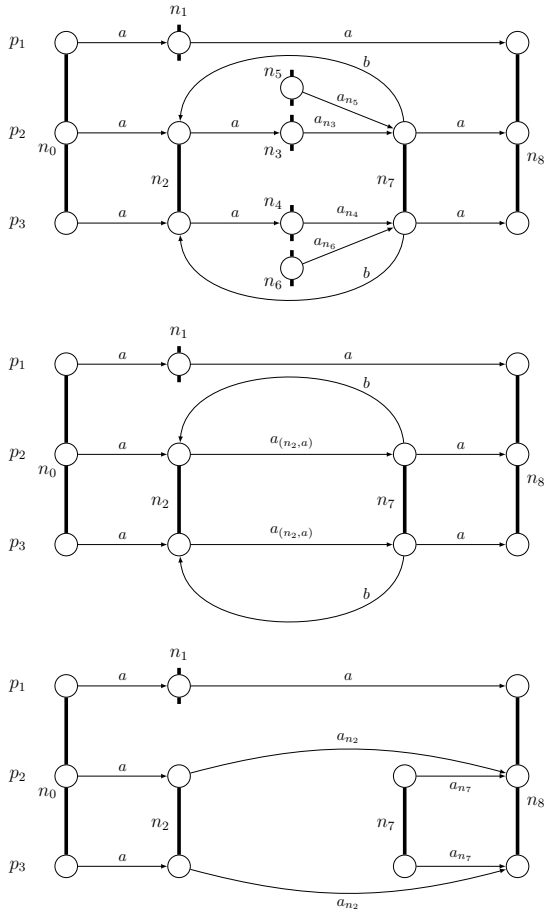
\begin{figure}[ht]
\centerline{\scalebox{0.60}{\def\hs{2.5}
\def\vs{2.0}
\begin{tikzpicture}
\vnego[ports=3,id=n0,spacing=\vs]{0,0}
\node[above left = 1.2cm and -0.1 cm of n0_P0, font=\large] {$n_0$};
\node[left = 0.5cm of n0_P2, font=\large] {$p_1$};
\node[left = 0.5cm of n0_P1, font=\large] {$p_2$};
\node[left = 0.5cm of n0_P0, font=\large] {$p_3$};
\vnego[ports=1,id=n1]{\hs,2*\vs}
\node[above = 0.2cm  of n1_P0, font=\large] {$n_1$};
\vnego[ports=2,id=n2,spacing=\vs]{\hs,0}
\node[above left = 0.5cm and -0.1cm of n2_P0, font=\large] {$n_2$};
\vnego[ports=1,id=n3]{2*\hs,\vs}
\node[below left = 0cm and -0.1cm of n3_P0, font=\large] {$n_3$};
\vnego[ports=1,id=n4]{2*\hs,0.0}
\node[above left = 0cm and -0.1cm of n4_P0, font=\large] {$n_4$};
\vnego[ports=1,id=n5]{2*\hs,1.5*\vs}
\node[above left = -0.1cm and -0.1cm of n5_P0, font=\large] {$n_5$};
\vnego[ports=1,id=n6]{2*\hs,-0.5*\vs}
\node[below left = -0.1cm and -0.1cm of n6_P0, font=\large] {$n_6$};
\vnego[ports=2,id=n7,spacing=\vs]{3*\hs,0}
\node[above left = 0.5cm and -0.1cm of n7_P0, font=\large] {$n_7$};
\vnego[ports=3,id=n8,spacing=\vs]{4*\hs,0}
\node[above right = 1.0cm and -0.1 cm of n8_P0, font=\large] {$n_8$};

\pgfsetarrowsend{latex}
\draw (n0_P2) -- (n1_P0) node [above,midway]{$a$};
\draw (n0_P1) -- (n2_P1) node [above,midway]{$a$};
\draw (n0_P0) -- (n2_P0) node [above,midway]{$a$};
\draw (n1_P0) -- (n8_P2) node [above,midway]{$a$};
\draw (n2_P1) -- (n3_P0) node [above,midway]{$a$};
\draw (n2_P0) -- (n4_P0) node [above,midway]{$a$};
\draw (n3_P0) -- (n7_P1) node [above,near start]{$a_{n_3}$};
\draw (n4_P0) -- (n7_P0) node [above,midway]{$a_{n_4}$};
\draw (n5_P0) -- (n7_P1) node [above,midway]{$a_{n_5}$};
\draw (n6_P0) -- (n7_P0) node [above,near start]{$a_{n_6}$};
\draw (n7_P1) -- (n8_P1) node [above,midway]{$a$};
\draw (n7_P0) -- (n8_P0) node [above,midway]{$a$};
\draw (n7_P1) to [bend right = 85] node [above,near start]{$b$} (n2_P1);
\draw (n7_P0) to [bend left = 85] node [above,near start]{$b$} (n2_P0);

\end{tikzpicture}}}
\centerline{\scalebox{0.60}{\def\hs{2.5}
\def\vs{2.0}
\begin{tikzpicture}
\vnego[ports=3,id=n0,spacing=\vs]{0,0}
\node[above left = 1.2cm and -0.1 cm of n0_P0, font=\large] {$n_0$};
\node[left = 0.5cm of n0_P2, font=\large] {$p_1$};
\node[left = 0.5cm of n0_P1, font=\large] {$p_2$};
\node[left = 0.5cm of n0_P0, font=\large] {$p_3$};
\vnego[ports=1,id=n1]{\hs,2*\vs}
\node[above = 0.2cm  of n1_P0, font=\large] {$n_1$};
\vnego[ports=2,id=n2,spacing=\vs]{\hs,0}
\node[above left = 0.5cm and -0.1cm of n2_P0, font=\large] {$n_2$};
\vnego[ports=2,id=n7,spacing=\vs]{3*\hs,0}
\node[above left = 0.5cm and -0.1cm of n7_P0, font=\large] {$n_7$};
\vnego[ports=3,id=n8,spacing=\vs]{4*\hs,0}
\node[above right = 1.0cm and -0.1 cm of n8_P0, font=\large] {$n_8$};

\pgfsetarrowsend{latex}
\draw (n0_P2) -- (n1_P0) node [above,midway]{$a$};
\draw (n0_P1) -- (n2_P1) node [above,midway]{$a$};
\draw (n0_P0) -- (n2_P0) node [above,midway]{$a$};
\draw (n1_P0) -- (n8_P2) node [above,midway]{$a$};
\draw (n2_P1) -- (n7_P1) node [above,midway]{$a_{(n_2,a)}$};
\draw (n2_P0) -- (n7_P0) node [above,midway]{$a_{(n_2,a)}$};
\draw (n7_P1) -- (n8_P1) node [above,midway]{$a$};
\draw (n7_P0) -- (n8_P0) node [above,midway]{$a$};
\draw (n7_P1) to [bend right = 85] node [above,near start]{$b$} (n2_P1);
\draw (n7_P0) to [bend left = 85] node [above,near start]{$b$} (n2_P0);

\end{tikzpicture}}}
\centerline{\scalebox{0.60}{\def\hs{2.5}
\def\vs{2.0}
\begin{tikzpicture}
\vnego[ports=3,id=n0,spacing=\vs]{0,0}
\node[above left = 1.2cm and -0.1 cm of n0_P0, font=\large] {$n_0$};
\node[left = 0.5cm of n0_P2, font=\large] {$p_1$};
\node[left = 0.5cm of n0_P1, font=\large] {$p_2$};
\node[left = 0.5cm of n0_P0, font=\large] {$p_3$};
\vnego[ports=1,id=n1]{\hs,2*\vs}
\node[above = 0.2cm  of n1_P0, font=\large] {$n_1$};
\vnego[ports=2,id=n2,spacing=\vs]{\hs,0}
\node[above left = 0.5cm and -0.1cm of n2_P0, font=\large] {$n_2$};
\vnego[ports=2,id=n7,spacing=\vs]{3*\hs,0}
\node[above left = 0.5cm and -0.1cm of n7_P0, font=\large] {$n_7$};
\vnego[ports=3,id=n8,spacing=\vs]{4*\hs,0}
\node[above right = 1.0cm and -0.1 cm of n8_P0, font=\large] {$n_8$};

\pgfsetarrowsend{latex}
\draw (n0_P2) -- (n1_P0) node [above,midway]{$a$};
\draw (n0_P1) -- (n2_P1) node [above,midway]{$a$};
\draw (n0_P0) -- (n2_P0) node [above,midway]{$a$};
\draw (n1_P0) -- (n8_P2) node [above,midway]{$a$};
\draw (n7_P1) -- (n8_P1) node [below,midway]{$a_{n_7}$};
\draw (n7_P0) -- (n8_P0) node [above,midway]{$a_{n_7}$};
\draw (n2_P1) to [bend left = 15] node [above,midway]{$a_{n_2}$} (n8_P1);
\draw (n2_P0) to [bend right = 15] node [below,midway]{$a_{n_2}$} (n8_P0);

\end{tikzpicture}}}
\caption{\small Three reduced negotiation diagrams constructed by Algorithm
  \ref{alg:red} started on the negotiation diagram of Figure \ref{fig:decomp}.}
\label{fig:intermediate}
\end{figure}

Finally, to estimate the complexity of the algorithm, we should also
examine how to 
calculate $\Redl(\Nn)$ and $\Redn(\Nn)$ in lines~$7$ and~$14$ of the algorithm.
In order to calculate $\Redl(\Nn)$ we need to know $F(\ell)$. 
The invariant says that $\Nnl$ at that point has one trace.
So it is enough to execute this trace in $\Nnl$ to reach $F(\ell)$.
Similarly, for $\Redn(\Nn)$ we need to know $F(n)$.
The invariant says that $\Nn_n$ is a replication, so it is just a flow
graph with one final node $n_\fin$. We can  execute in $\Nn$ any path leading 
to the exit to calculate $F(n)$. 
To sum up, the calculations of $\Redl(\Nn)$ and $\Redn(\Nn)$ in
lines~$7$ and~$14$ can be done in linear time with respect to the size
of $\Nn$.

We summarize the results of this section:

\begin{theorem}\label{thm:main}
  Let $\Nn$ be a sound deterministic negotiation diagram, and $\sem{\_}$ a
  Mazurkiewicz invariant analysis framework. Algorithm~\ref{alg:red}
  stops and outputs $\sem{\Nn}$. The complexity of the algorithm is
  $\Oo(|\Nn|(C+|\Nn|))$ where $|\Nn|$ is the size of $\Nn$, and $C$ is
  the cost of $\sem{\_}$ analysis for flow-graphs of  size $|\Nn|$.  
\end{theorem}

%%% Local Variables: 
%%% mode: latex
%%% TeX-master: "../mlics"
%%% End: 

\section{Anti-Patterns and Gen/Kill Analyses}
\label{sec:genkill}

We saw in Section \ref{subsec:frameworks} that the problem of detecting an anti-pattern
can be naturally captured as an analysis framework, which however is not
Mazurkiewicz-invariant. We now show that, while the {\em natural} framework is not 
Mazurkiewicz invariant, an equivalent framework that returns the same result is. 
Then we sketch how this result generalizes to arbitrary Gen/Kill analysis
frameworks, a much studied class  \cite{Hecht77}.

We need some basic notions of Mazurkiewicz trace theory~\cite{book-of-traces}.
Given a sequence $w=w_1 \cdots w_n \in\Ll^\star$, define $i\tleq' j$ for two
positions $i,j$ of $w$ if $i\leq j$ and $w_i$ is not independent
from $w_j$ (see Definition~\ref{def:mazur}). Further, define
$\tleq$ is the transitive closure of the relation $\tleq'$. It is well-known
that if $w \equiv v$, i.e., if $w,v$ are
Mazurkiewicz equivalent, then $\tleq_w$ and $\tleq_v$ 
are isomorphic as labeled partial orders (the labels being the locations). 
We write $\bet(i,j)$ for the set of positions between $i$ and $j$, i.e., 
$\set{k : i\tleq k \tleq j}$. 

Recall that anti-pattern analysis asked if there is an execution with
the property ``$w\in L$'' for the language $L=\Ll^*\ell_1(\bK)^*\ell_2\Ll^*$.
Instead of this property consider the following property of $w$:
 \begin{quote}
  (*) there are two positions $i,j$ such that $w_i=\ell_1$,
  $w_j=\ell_2$,  $\bet(i,j)\cap K=\es$, and   not $j\tleq i$.
\end{quote}
The special case of this condition is when $\bet(i,j)=\es$; then, since
not  $j\tleq i$, we actually know that the two positions $i,j$ are
concurrent, so $w$ is Mazurkiewicz equivalent to $w_1\ell_1\ell_2 w_2$. 

\begin{lemma}
\label{lem:traceform}
  A negotiation diagram $\Nn$ has a successful run $w\in L$ iff it has
  a successful run $v$ with  property (*).
\end{lemma}

It remains to set a static analysis framework for tracking 
property (*).  Since this is a property of Mazurkiewicz traces, the framework
is Mazurkiewicz invariant. 

We need one more piece of notation. For a word $w$ and a process
$p$ we write $\bet(\ell_1,p)$ for the set $\bet(i,j)$ where $i$ is
the last occurrence of $\ell_1$, and $j$ is the last occurrence of a
location using process $p$.
%: $i$ is the biggest index with $w_i=\ell_1$,
%and $j$ is the biggest index 
%with $w_j=\ell$ for some $\ell$ having $p$ in its domain. 

We define now an auxiliary function $\a$ from sequences to
$\Pp(\Proc)^2\cup\set{\top}$. 
We set $\a(w)=\top$ if $w$ has property (*), otherwise
$\a(w)=(P_A,P_B)$ where
\begin{itemize}
\item $p\in P_A$ if $\bet(\ell_1,p)\not=\es$ and $\bet(\ell_1,p)\cap K=\es$,
\item $p\in P_B$ if $\bet(\ell_1,p)\not=\es$ and $\bet(\ell_1,p)\cap K\not=\es$.
\end{itemize}

%\item $p\in P_C$ if $\bet(\ell_1,p)=\es$.
%\end{itemize}
% Observe that $\a(w)$ is either $\top$ or a partition of the set of processes in
% three sets. 

We describe a \PTIME\ computable function $F$ such
that for every sequence $w$ and location $\ell$:
\begin{equation*}
  \a(w\,\ell)=F(\a(w),\ell)
\end{equation*}
For defining $F$ we first describe the update of
$\bet(\ell_1,p)$ when extending $w$ by $\ell$. Let us detail the more
interesting case where $\ell \not=\ell_1$. Observe that the set of
positions $\bet(\ell_1,p)$ does not change if $p \notin
\dom(\ell)$. If $p \in\dom(\ell)$ then the update of $\bet(\ell_1,p)$ is
the union of $\bet(\ell_1,q)$ over all $q \in \dom(\ell)$, plus
$\ell$. According to these observations, we define 
the update of $F$ in the case $\ell \not=\ell_1$ goes as follows. If 
$p \notin \dom(\ell)$ then
process $p$ remains in its set, $P_A$ or $P_B$. If $p
\in\dom(\ell)$, then: $p$ goes into the set $P'_B$ if 
either there was some $q \in\dom(\ell)$ in $P_B$, or $\ell \in
K$ and there is some $q \in\dom(\ell)$ in $P_A$; $p$ goes into the set 
$P'_A$ if $\ell \notin K$, no $q \in\dom(\ell)$ is in $P_B$ and there 
is at least one $q \in\dom(\ell)$ in $P_A$.  

The function $F$  can be extended to a monotone and distributive
function $\wh F$ on 
% %Moreover the function $F$ is \PTIME\ computable.
% (It can also be extended to a distributive monotone \anca{why is it monotone?} function
$\Pp(\Proc)^3\cup\set{\top}$ ordered componentwise, by turning $\a(w)$
into a  partition of $\Proc$ (adding a component $P_C=\Proc \setminus
(P_A \cup P_B)$) and embedding it into a suitable function
over $\Pp(\Proc)^3\cup\set{\top}$. 

With the help of the function $\wh F$ we define now the value of each location:
\begin{equation*}
  \sem{\ell}(P_A,P_B,P_C)=\wh F((P_A,P_B,P_C),\ell)
\end{equation*}
Observe that this gives us $\sem{w}=\a(w)$ for every sequence $w$. 
The above discussion yields two lemmas showing that $\sem{\_}$ is a
Mazurkiewicz invariant analysis framework that can be computed in
\PTIME, since $\wh F$ can be computed in \PTIME.

\begin{lemma}
  $\sem{\cdot}$ is Mazurkiewicz-invariant.
\end{lemma}

\begin{lemma} 
  %Let $\Proc$ be a fixed set of processes. 
Consider a negotiation diagram $\Nn$ 
  over set of locations $\Ll$. 
For every sequence $w \in \Ll^*$: 
 $\sem{w}(\es,\es,\Proc)=\top$ iff $w\in L$. 
 Moreover,
 $\sem{\Nn}(\es,\es,\Proc)=\top$ iff $\Nn$ has a successful 
    execution in $L$.
\end{lemma}

\subsection{Generalization to Gen/Kill analyses}
We consider general Gen/Kill analyses\footnote{Although not in bitvector form,
which we leave for future work (see the conclusions).}. We are given a set
of locations $G\incl \Ll$ that \emph{generate} something, and a set of locations
$K\incl \Ll$ (not necessarily disjoint with $G$) that \emph{kill} this something. The
lattice $\Dd$ has just two elements $\{0,1\}$, with $\wedge$ and $\vee$ as lattice operations,
and the transformer of a program instruction 
$\ell$ is of the form $\sem{\ell}(v) = (v \wedge (\ell \not\in K)) \vee (\ell \in G)$.
Classical examples from the static analysis of programs are the computation of reaching definitions, available expressions, live variables, very busy expressions, where the ``something''
are values assigned to a variable or an expression \cite{nielson}. 
The four main classes of Gen/Kill analyses differ
only on whether control-flow is interpreted forward or backwards, and
on whether we do ``merge over all paths'' or ``meet over all paths''.

\begin{itemize}
\item \textbf{may/forward}. For some configuration $C$ there is an  execution
  $C_\init\act{w}  C$ with $w\in\Ll^*G(\bar K)^*\ell$.
\item \textbf{must/forward}. For every configuration $C$ and every execution
  $C_\init\act{w} C$, if  $w$ ends with $\ell$ then $w\in \Ll^*G(\bK)^*\ell$.
\item \textbf{may/backward}. For some reachable configuration $C$ there is an 
  execution $C\act{w} C_\fin$ with  $w\in \ell(\bK)^*G\Ll^*$.
\item \textbf{must/backward}. For every reachable configuration $C$ and every execution
  $C\act{w} C_\fin$, if $w$ starts with $\ell$ then $w\in \ell(\bK)^*G\Ll^*$.
\end{itemize}

Observe that in backward properties we require that a configuration $C$
is reachable from the initial configuration. For forward properties
we do not need to assume that a configuration is co-reachable from the
final configuration, as this will be immediately implied by
soundness. 

The two existential properties above can be
expressed in terms of the existence of successful executions of a
particular form: executions $C_\init\act{w} C_\fin$ with $w$ belonging
to some language.  The same is true for the negation of the universal properties. 
Consider the languages given by regular expressions:
\begin{enumerate}
\item $E_1=\Ll^* \, G(\bK)^* \ell \,\Ll^*$,
\item $E_2=(\bK\cap\bG)^* \ell \,\Ll^*\cup \Ll^*(K\cap\bG)(\bK\cap\bG)^* \ell\,\Ll^*$,
\item $E_3=\Ll^*\, \ell(\bK)^*G\,\Ll^*$,
\item $E_4=\Ll^* \,\ell(\bK\cap\bG)^*\cup \Ll^*\, \ell(\bK\cap\bG)^*(K\cap\bG)\,\Ll^*$.
\end{enumerate}

\begin{lemma}\label{lem:reduction}
For a sound negotiation diagram we have the following:
  \begin{itemize}
  \item may/forward is equivalent to $\exists C_\init\act{w} C_\fin$
    with $w\in E_1$.
  \item negation of must/forward is equivalent to  $\exists C_\init\act{w} C_\fin$
    with $w\in E_2$. 
  \item may/backward is equivalent to $\exists C_\init\act{w} C_\fin$
    with $w\in E_3$.
  \item negation of must/backward is equivalent to $\exists C_\init\act{w} C_\fin$
    with $w\in E_4$. 
  \end{itemize}
\end{lemma}

The resource analysis at the beginning of this section corresponds to $E_3$. For each one
of $E_1, E_2, E_4$ it is easy to produce an analogon of Lemma \ref{lem:traceform} reformulating the property in trace terms. This allows us to check all properties in polynomial time using our algorithm for Mazurkiewicz invariant analysis frameworks.

\section{Conclusions}\label{sec:concl}

Previous work had identified deterministic negotiations -- a model of concurrency 
essentially isomorphic to free-choice 
workflow Petri nets -- as a class that has both  practical relevance for 
business process modeling, and admits \PTIME\ analysis for several
important properties once negotiations are assumed to be sound.
Moreover soundness is a natural prerequisite that can be checked in
\PTIME.
% and for which
% several important analysis problems can be solved in \PTIME.
% More precisely, in several papers it has been shown that 
% the soundness problem can be solved in \PTIME, and then that the same holds 
% for several properties of negotiations that are both deterministic and sound. 
%This is interesting both in theory and practice: unsound workflows have deadlocks, and so 
%soundness is a very desirable property, and these results show that a semantic property
%that is desirable in itself makes the verification problem for other properties easier!. 

We have proposed a general notion of Mazurkiewicz-invariant analysis frameworks.
%that can express all properties known to be effectively analyzable for sound
%deterministic negotiations. 
We have shown that computing the MOP in such frameworks for sound
deterministic negotiations is as easy as computing it for sequential
flow graphs (while computing the MOP of general frameworks takes
exponentially longer, unless \PTIME=\NP). 
This result not only subsumes all previous \PTIME\ results on analysis
of sound deterministic negotiations, but also yields \PTIME\
algorithms for new problems, like the computation of the
best-case/worst-case execution time, the detection of anti-patterns,
and general gen/kill analysis problems.
%, both previously discussed in
%the literature. 
The result is particularly interesting for gen/kill problems: While
their natural formulation is not in terms of Mazurkiewicz-invariant
frameworks, we have shown that they can be reformulated as
such. 
%Summarizing, we have presented the first 
%framework for the static analysis of negotiations/workflow Petri
%nets, and drastically extended the range of efficiently checkable properties.

In future work we plan to improve the degree of the polynomial
bounding the runtime of our algorithm.
Since our decomposition does not partition a negotiation into disjoint parts,
when computing MOPs of subnegotiations we are redoing computations.
Bounding the size of overlaps looks like a promising way of bringing
the complexity on a par with the sequential case. Section~\ref{sec:genkill} on gen/kill analyses raises a further question. In the sequential case, a gen/kill analysis can be simultaneously computed for all program points and all program variables (for example, one can compute for each program point the set of live variables at that point). This is not yet the case in our algorithm.
In fact, it is not clear what is a program point in a negotiation: If
one takes a configuration as a program point, then, since the number
of reachable configurations can grow exponentially in the size of the
negotiation diagram, any algorithm that explicitly computes the MOP
for each reachable configuration has exponential worst-case
complexity. 

\medskip 

\noindent \textbf{Acknowledgments.} We thank the anonymous reviewers for useful remarks, and J\"org Desel, Denis Kuperberg, and Philipp Hoffmann for helpful discussions.

%%% Local Variables:
%%% mode: latex
%%% TeX-master: "mlics"
%%% End:

\bibliographystyle{abbrv}
\bibliography{references}
\newpage
\appendix
\section*{Proofs from Section~\ref{subsec:mazur}}

\lemmaMazur*
\begin{proof}
It is easy to see that if $\ell_1$ and $\ell_2$ are independent and $C
\act{\ell_1\ell_2} C'$ for some configurations $C, C'$, then $C
\act{\ell_2\ell_1} C'$.  It follows that the same holds for any two  
Mazurkiewicz equivalent runs $v, w$. Indeed, if $w$ is successful, then by definition $C_\init \act{w} C_\fin$. Hence $C_\init \act{v} C_\fin$, and so $v$ is also successful.
\end{proof}

\lemmaMazurii*
\begin{proof} 
Let $w$ be a successful run, i.e., $C_{\it init} \act{w} C_{\it
fin}$. 
Observe that, since $\Nn$ is deterministic and the domain of $n_\fin$
contains all processes, $n_{\it fin}$ is the only node enabled 
at $C_{\it fin}$. 

We first prove that \emph{at most} one run $v\equiv w$ is
compatible with $S$. 
Assume there are two different such runs $v_1, v_2$. 
Then $v_1 = u \ell_1 u_1$ and $v_2 = u \ell_2
u_2$ for some $u, u_1, u_2$ and $\ell_1 \neq \ell_2$. Since $v_1 \equiv
v_2\equiv w$, the locations $\ell_1$ and $\ell_2$ are
independent. 
But then $S(u) \neq \ell_1$ or $S(u) \neq \ell_2$, and so at least one
of the two runs is not compatible with $S$.  

We now construct a run $v\equiv w$ that is compatible with $S$.
Suppose $w\equiv v_1w_2$ and $v_1$ is compatible with $S$. 
If $w_2$ is not empty then we will show how to prolong $v_1$ to an $S$
compatible run $v'_1$, and 
shorten $w_2$ to $w'_2$ so that still $w\equiv v'_1w'_2$.
Let $S(v_1)=n$. 
If we look at the run $C_\init\act{v_1}C_1\act{w_2}C_\fin$ then $n$ is
enabled in $C_1$. 
If $n=n_\fin$ then $w_2$ is the empty sequence and we are done.
Otherwise, since $n$ is enabled in $C_1$, $w_2$ must be of the form
$u\ell u'$ where  no literal in $u$ uses a process from
$\dom(n)$, and $\ell=(n,b)$ for some outcome $b$. 
But then $w_2\equiv \ell uu'$, and $w\equiv v_1\ell uu'$.
We have found our desired $v'_1=v_1\ell$ and $w'_2=uu'$.
\end{proof}

\begin{example}
Observe that Lemma \ref{lem:mazur2} may not hold for runs that are
not successful nor for non-deterministic negotiation
diagrams. We explain this second point in more detail.
Consider the diagram of Figure \ref{fig:sched} with $n_0$ and $n_4$ as initial and final 
nodes (the hyperarcs indicate $\delta(n_0,p_1,a)= \{n_1, n_2\}$ and 
$\delta(n_0,p_1,a)= \{n_2, n_3\}$. Consider the scheduler that after the occurrence of 
$(n_0, a)$ selects the node $n_1$. The successful trace $[(n_0,a) (n_2, a)]$ 
does not have any run compatible with this scheduler.\end{example}
\begin{figure}[ht]
\centerline{\scalebox{0.8}{\begin{tikzpicture}
\vnego[ports=2,id=n0,spacing=2.0]{0,0}
\node[above left = 0.5cm and -0.1 cm of n0_P0, font=\large] {$n_0$};
\vnego[ports=1,id=n1]{2.0,2.5}
\node[above left = -0.1cm and 0cm of n1_P0, font=\large] {$n_1$};
\vnego[ports=2,id=n2,spacing=1.0]{2.0,0.5}
\node[above left = 0cm and 0cm of n2_P0, font=\large] {$n_2$};
\vnego[ports=1,id=n3]{2.0,-0.5}
\node[below left = -0.1cm and 0cm of n3_P0, font=\large] {$n_3$};
\vnego[ports=2,id=n4,spacing=2.0]{4.0,0}
\node[above left = 0.5cm and -0.1cm of n4_P0, font=\large] {$n_4$};
\node[left = 0.5cm of n0_P1, font=\large] {$p_1$};
\node[left = 0.5cm of n0_P0, font=\large] {$p_2$};

\pgfsetarrowsend{latex}
\draw (n0_P1) -- (1,2) node [above,midway]{$a$} -- (1,2) -- (n1_P0);
\draw (1,2) -- (n2_P1);
\draw (n0_P0) -- (1,0) node [above,midway]{$a$} -- (1,0) -- (n2_P0);
\draw (1,0) -- (n3_P0);
\draw (n1_P0) -- (n4_P1)  node [above,midway]{$a$};
\draw (n2_P0) -- (n4_P0)  node [above,midway]{$a$};
\draw (n2_P1) -- (n4_P1)  node [below,midway]{$a$};
\draw (n3_P0) -- (n4_P0)  node [below,midway]{$a$};
\end{tikzpicture}}}
\caption{\small A non-deterministic negotiation diagram for which Lemma \ref{lem:mazur2} does not hold. }
\label{fig:sched}
\end{figure}
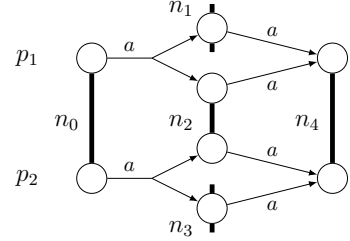

\thmMazur*
\begin{proof}
By Lemma \ref{lem:mazur2}, there exists a bijection $\phi$ between the successful runs 
compatible with $S$ and $S'$: Given a successful run $w$ compatible with $S$, we define
$\phi(w)$ as the unique run such that $w \equiv \phi(w)$. Since $\sem{\_}$ is Mazurkiewicz invariant, we further have $\sem{w} = \sem{\phi(w)}$. Letting $w$ and $w'$ range over the successful runs of $\Nn$ compatible with $S$ and $S'$, respectively we obtain
$\sem{\Nn, S} = \bigsqcup_w \sem{w} = \bigsqcup_w \sum{\phi(w)} = \bigsqcup_{w'} \sem{w'} = \sem{\Nn, S'}$.
\end{proof}

\section*{Proofs from Section~\ref{sec:decomp}}
We recall some definitions and two results, one from \cite{negII} and one from from \cite{DBLP:conf/concur/EsparzaKMW16}. A local path $n_0\act{p_0,a_0}n_1 \cdots n_{k-1}\act{p_{k-1},a_{k-1}} n_k$ is \emph{realizable} from a configuration $C$ if there is a run $C\act{w}$ with $w$ of the form $(n_0,a_0) w_1 (n_1,a_1)\cdots w_{k-1} (n_{k-1},a_{k-1})$, such that $p_i$ does not belong to the domain of any location of $w_{i+1}$. 

\begin{lemma}{(\cite{DBLP:conf/concur/EsparzaKMW16})}
\label{lem:exec}
Let $C$ be a reachable configuration of a sound deterministic
negotiation diagram $\Nn$.
Every local path whose initial node is enabled at $C$ is realizable from $C$.
\end{lemma}

\begin{lemma}{(\cite{negII})}
\label{lem:sync}
Let $C$ be a reachable configuration of a sound deterministic
negotiation diagram $\Nn$
and let $C \act{w} C$ where $w = \ell_1 \cdots \ell_k$, $k>0$. There is an index
$1 \leq i \leq k$ such that $\ell_j\domleq \ell_i$ for every 
$1 \leq j \leq k$.
\end{lemma}

No we are ready to prove domination lemma restated below.

\lemmaDomin*
\begin{proof}
Let $\pi$ be a reachable local circuit of $\Nn$. 
By Lemma \ref{lem:exec}, every number of iterations of $\pi$ is
realizable. 
If we take a sufficiently big number of iterations, say $l$, and a reachable
configuration $C_1$, then we get an execution 
$C_1\act{u}C\act{v}C\act{u'}C_2$ with $uvu'$ realizing $\pi^l$.
We have that the looping part $v$ is of the form
 $v=(n_0,a_0) w_1 (n_1,a_1)\cdots w_{k-1} (n_{k-1},a_{k-1})$ 
where $n_0\act{p_0,a_0}n_1\act{p_1,a_1}\dots\act{p_{k-1},a_{k-1}}$ is
a, continues, subsequence of $\pi^l$.
Moreover realizability guarantees us that  $p_i$ does not belong to
the domain of  any location of $w_{i+1}$.
% So let  $w=(n_0,a_0) w_1 (n_1,a_1)\cdots w_{k-1}
% (n_{k-1},a_{k-1})$ be such that $p_i$ does not belong to the domain of
% any location of $w_{i+1}$ and $C_1 \act{w}$ for some reachable
% configuration $C_1$. 
% By iterating $\pi$ and since $\N$ has finitely many
% reachable configurations, this gives a loop $C \act{v} C$ of $\N$
% with $v \in w^+$. % containing all of $(n_1, a_1), \ldots, (n_{k-1}, a_{k-1})$. 
% %and no other location of $n_1, \ldots, n_{k-1}$. 
By Lemma~\ref{lem:sync} 
some location $\ell=(n,a)$ of $v$ satisfies $\dom(\ell') \incl
\dom(\ell)$ for every location $\ell'$ of $v$. 

We claim that $\ell= (n_j, a_j)$ for some $1 \leq j \leq k$, which
implies that $n_j$ is a dominating node of $\pi$. Let $C'$ be a
configuration reached during the execution of the loop $C \act{v} C$, that enables
$\ell$. Let $(n_i, a_i)$ be the last location of the path $\pi$
occurring in the loop before $C'$. By the definition of $\ell$ and of
$C'$ we have $C'(p) = n$ for every $p \in \dom(n_i)$, and so in
particular $C'(p_i)=n$. Since $p_i$ does not belong to the domain of
any location of $w_{i+1}$, %  contain any other location
% of $n_i$ but $(n_i, a_i)$ and $n_i \act{a_i, p_i} n_{i+1}$,
we also
have $C'(p_i)= n_{i+1}$. So $n = n_{i+1}$, hence the claim is
proved by taking $j =i+1$. 
\end{proof}

Next we prove the unique configuration lemma, that is the main
technical tool for the theorem that follows.

\lemmaMainTechnical*
\begin{proof}
  First suppose that there are nodes $m_1$ and $m_2$
  such that $m_1$ is enabled in $C_1$ but not in $C_2$, and
  symmetrically $m_2$ is enabled in $C_2$ but not in $C_1$.
  In particular, $m_1 \not= m_2$. Consider $m_1$.
  Since $\dom(m_1)\cap X\not=\es$ by  (2), we can find
$p_1\in X$ such that $C_1(p_1)=m_1$; hence also $C_2(p_1)=m_1$ by (1).
  As $m_1$ is not enabled in $C_2$, there is $q_1\in\dom(m_1)$ with
  $C_2(q_1)=n_1\not=m_1$. 
  Applying the same arguments to $m_2$ we get also $p_2$ and $q_2$,
  and have the following properties:
  \begin{align*}
    C_1(p_1)=&C_1(q_1)=m_1& C_2(p_1)=&m_1& C_2(q_1)=n_1\\
    C_2(p_2)=&C_2(q_2)=m_2& C_1(p_2)=&m_2& C_1(q_2)=n_2
  \end{align*}
  
  By soundness there is an execution from $C_2$ of the form
  \begin{equation*}
    C_2\act{m_2}\act{u}\act{n_1}\act{v}\act{m_1}\act{w}C_\fin
  \end{equation*}
  this is because $C_2(p_1)=m_1$, so $m_1$ needs to be executed at
  some point, and since $C_2(q_1)=n_1$ and $q_1\in \dom(m_1)$ then
  $n_1$ should be executed before $m_1$ can be executed.
  Observe that no node in $u$ or $v$ uses $p_1$,
  since $C_2(p_1)=m_1$ and $m_1$ is not executed in $u$ or $v$.
  This gives us a path from $m_2$ to $m_1$ in the graph of the
  negotiation diagram without any node using $p_1$. 
  The same argument gives us a path from $m_2$ to $m_1$ without any
  node using $p_2$.
  Since $p_1\not\in\dom(m_2)$ and $p_2\not\in\dom(m_1)$ we get a cycle
  in the graph of the negotiation diagram without a dominant node, contradicting
  Lemma \ref{lem:domin}.

  Now we consider the general case, and show that if $C_1\not=C_2$
  then we can get to the situation considered at the beginning of the
  proof.
  Take an execution from $C_1$ to the final configuration. 
  Let $u$ be the longest prefix of this execution that is possible
  to execute from $C_2$.
  We have
  \begin{equation*}
    C_1\act{u}C^u_1\act{m_1}\qquad C_2\act{u}C_2^u\not\act{m_1}
  \end{equation*}
  Observe that such a prefix must exist since $C_1\not=C_2$.
  Simply $m_1$ is the first node using a process $q$
  such that $C_1(q)\not=C_2(q)$.
  We claim that there is $p\in\dom(m_1)$ with $C^u_1(p)=C^u_2(p)$.
  If $\dom(m_1)\cap\dom(u)\not=\es$ then we can take any $p$ in this
  intersection.
  Otherwise $C_1\sat m_1$, and by (2) we have $dom(m_1)\cap X\not=\es$.
  So for $p$ in this intersection we get $C^u_1(p)=C_1(p)=C_2(p)=C^u_2(p)$.

  Since $m_1$ is not enabled in $C^u_2$, there is process
  $q\in\dom(m_1)$ with $C^u_2(q)=n_1\not=m_1$. 
  Since $C^u_2(p)=m_1$, there is an execution from $C^u_2$ reaching
  $m_1$:
  \begin{equation*}
    C^u_2\act{w}\act{n_1}\act{w'}\act{m_1}\dots
  \end{equation*}
  Observe that $n_1$ must appear before $m_1$ on this execution. 
  Consider the longest  prefix $v$ of this execution that is possible
  from $C^u_1$. 
  We have that $v$ is a prefix (possibly not strict) of $w$, since it
  is not possible to execute $n_1$ from $C^u_1$ before executing
  $m_1$.
  We obtain
  \begin{equation*}
    C^u_2\act{v}C^{uv}_2\act{m_2}\qquad C^u_1\act{v}C^{uv}_1\not\act{m_2}
  \end{equation*}
  Finally, observe that $C^{uv}_1\act{m_1}$ since $C^u_1\act{m_1}$.
  Moreover, $C^{uv}_2\not\act{m_1}$ since $C^{uv}_2(q)=n_1$, and $m_1$
  is does not appear in $v$. 
  Thus $C^{uv}_1$ and $C^{uv}_2$ are configurations satisfying our
  assumption from the first paragraph.
  Since the first paragraph shows that such two different configurations
  cannot exist, we get $C_1=C_2$.
\end{proof}

\thmUnique*
\begin{proof}
For (i), take $X=\dom(m)$. Suppose 
that there are two configurations $I_1$ and $I_2$ as in (i). The hypotheses of Lemma~\ref{lemma:main-technical} are
satisfied, and so $I_1=I_2$.%  For (ii), take $X=\Proc\setminus\dom(m)$. 
% Since $C_1$ and $C_2$ are obtained from $\first{m}$ by a
% $\dom(m)$-sequence, we have $C_1(X)=C_2(X)$, and 
% Lemma~\ref{lemma:main-technical} gives $C_1=C_2$.

For (ii) take $X=\Proc\setminus\dom(m)$. Suppose 
that there are two configurations $F_1$ and $F_2$ as in (ii). 
We get that the two configurations agree on $X$ since they are
obtained from $I(m)$ by $m$-sequences. By definition of the two $F$
configurations, every node enabled in
one of $F_1$, $F_2$ needs a process from $X$, so $F_1=F_2$
by Lemma~\ref{lemma:main-technical}.

For (iii) we take $m$ and $\ell$ as in the statement.
We prove a stronger property by induction. 
Let $C$ be a reachable configuration of $\Nn$, and $v,w$ two strict
$m$-sequences leading to $F_v(\ell)$ and $F_w(\ell)$, respectively,
satisfying the two conditions of (iii). 
We want to show that $F_v(\ell)=F_w(\ell)$. 
The proof is by induction on the sum of the lengths of $v$ and $w$.
The conclusion  then follows by taking as $C$ the configuration
obtained from $I(m)$ after doing $\ell$.

For the induction step, let us take the first location $(n,a)$ of $v$,
i.e. $v=(n,a)v'$. 
Now $w$ must be Mazurkiewicz equivalent to $(n,a)w'$, since $n$ is
enabled in $C$. 
Thanks to Lemma~\ref{lem:mazur},we obtain two computations
$C\act{(n,a)}C'\act{v'}F_v(\ell)$, and 
$C\act{(n,a)}C'\act{w'}F_w(\ell)$.
By induction hypothesis we have $F_v(\ell)=F_w(\ell)$.
\end{proof}

Before proving that $\Nnn$ is sound we need an important technical
lemma. 

\begin{restatable}{lemma}{lemmaMnotInside}
\label{lemma:m-not-inside}
  If $m=F(n)(p)$ for some process $p$ then there is no reachable
  configuration of $\Nnn$ where $m$ is enabled.
\end{restatable}
\begin{proof}
Suppose by contradiction that $m$ is executed in $\Nnn$ and
$F(n)(p)=m$ for some process $p$. So we can find an $n$-sequence
$u$ such that  
\begin{equation*}
  I(m)\act{u} F(n)
\end{equation*}
is an execution of $\Nn$. We choose $u$ minimal, so that $m$ is
executed only once, at the beginning of $u$. 
Since in  $F(n)$ no node $m' \preceq n$ is
enabled, we have $F(n)(p_1)=m_1\not=m$ for some $p_1\in\dom(m)$. 

If $m_1 \preceq n$, then since $m_1$ is not enabled
in $F(n)$ we have $F(n)(p_2)=m_2$ for some $p_2 \in
\dom(m_1)$. By induction we get a sequence of processes
$p_1,\ldots,p_k$ and of nodes $m_1,\ldots, m_k$ such that:
\begin{itemize}
\item $F(n)(p_i)=m_i$ and $p_i \in\dom(m_{i-1})$ for all $i$,
\item $m_i \preceq n$ for all $i<k$, and $m_k \not\preceq n$.
\end{itemize}

Intuitively, each $p_i$ has to wait for $p_{i+1}$ in order to execute
$m_i$. Since $\Nn$ is sound there is some execution from $F(n)$ that enables
$m$. Let us consider such an execution
\[
C(m) \act{u} F(n) \act{u_k} C_k \act{u_{k-1}}
\cdots \act{u_1} C_1 \act{u_0} C
\]
such that
\begin{itemize}
\item $C \sat m$,
\item for all $i \le k$, $u_{i-1}$ starts with $\ell_i=(m_i,a_i)$ for some $a_i$,
\item process $p$ does not occur in $u_k
  \cdots u_1 u_0$.
\end{itemize}

Recall that $u$ starts by $(m,a)$ for some $a$. The above execution
yields some local path $\pi$ from $m$ to $m$ containing 
$m_k$. 
This path should have a dominant node by Lemma~\ref{lem:domin}.
Since $m_k \not\preceq n$ the set of processes occurring in
$\pi$ is not included in $\dom(n)$. 
So the dominant node cannot be a part of $u$ since the domains of
nodes in $u$ are included in $\dom(n)$.
The dominant node cannot be a part of $u_k\dots u_0$ either since 
process $p$ does not occur in $u_k \cdots u_0$. We
obtain thus a contradiction.
% Let $m'$ be the first node in
% $\pi$ involving some process $p' \notin \dom(n)$. Note that $m'$ is
% either $m_k$ or some node before (from $u_k$). 

% Since $\Nn$ is sound and $\pi$ is a local cycle in $\Nn$, there must
% be some dominant node on $\pi$. However, there can be no dominant node
% before $m'$, because that part contains only nodes with domain
% included in $\dom(n)$. Neither can there be a dominant node after
% $m'$, since process $p$ does not occur in $u_k \cdots u_0$. We
% obtain thus a contradiction.

% For simplicity assume that $\dom(m')\not\incl\dom(n)$;
% say $r\in \dom(m')\setminus \dom(n)$.
% We hope that the argument below extends to something that avoids this
% assumption. 

% Since $\Nn$ is is sound we prolong the above execution to 
% \begin{equation*}
%   C(m)\act{u} D(n) \act{v} C \act{m}
% \end{equation*}
% where $m$ is not executed on $v$.
% This execution gives a local path on process $q$:
% \begin{equation*}
%   m\act{u'}m'\act{v'}m
% \end{equation*}
% where $u'$ is a subsequence of $u$, and $v'$ is a subsequence of $v$.
% Since this path is a loop it should have a dominant action by
% Lemma~\ref{lemma:dominant}. But actions in $u'$ have domains inside
% $\dom(n)$ so they do not use $r\in\dom(m')$. 
% On the other hand, $m'$ as well as actions in $v'$, do not use $p$, since $p$
% does not execute on $v$.
% A contradiction.
\end{proof}

\lemmansound*
\begin{proof}
Consider a run $C_{\it init}|_n \act{v} C_1$ in $\Nn|_n$, where $C_{{\it init}}|_n$ is the initial
configuration of $\Nn|_n$, i.e., $C_{\it init}|_n(p)= \{n\}$ for every $p \in \dom(n)$. We prove that the run can be extended to as successful run 
of $\Nn|_n$.

Since $n$ is reachable in $\Nn$, by Theorem \ref{thm:unique}(i) we can take a run $C_{\it{init}}\act{u} \first{n}$ and prolong it to $C_{\it{init}} \act{u} \first{n} \act{v} C_1'$ such that (i) $C_1'(p)=\first{n}(p)$ for every process $p \notin \dom(n)$, and 
  (ii)  $C_1'(p)=C_1(p)$ or ($C_1(p)=n|^\fin_n$ and $C_1'(p)=\last{n}$) for every process $p \in \dom(n)$.  Since $\Nn$ is sound, we can prolong the run further to an accepting one, say $C_{{\it init}} \act{u} \first{n} \act{v} C_1'\act{w} C_{{\it fin}}$.
We now permute exhaustively consecutive independent outcomes $\ell \ell'$ in $w$ such that $\dom(\ell) \subseteq \dom(n)$ and $\dom(\ell')$ is not included in $\dom(n)$; say the result is $w_1w_2$. Then we have 
 $C_{{\it init}} \act{u} \first{n} \act{v} C_1'\act{w_1} C_2'$.
Lemma~\ref{lemma:m-not-inside} gives us then an execution  $C_{\it
  init}|_n \act{v} C_1 
\act{w_1} C_2$ for configurations $C_2$ and $C_2'$ satisfying the
conditions (i) and (ii) above  (that is, $C_2$ and $C_2'$ satisfy the
same tow conditions as $C_1$ and $C_1'$). 
Moreover, since the outcomes are permuted exhaustively, either $C_2'=C_{fin}$, or every node enabled in $C_2'$ needs a process outside of $\dom(n)$. By Theorem \ref{thm:unique}(ii), in both cases we have $C_2'=\last{n}$. By condition (ii) above, we have $C_2(p)= n|^\fin_n$ for every process $p \in \dom(n)$. So $C_2$ is the final configuration of $\Nn|_n$, and the run can be prolonged to a successful run. 
\end{proof}

\lemlsound*
\begin{proof}
Analogous to the proof of Lemma \ref{lem:nsound}, replacing 
Theorem \ref{thm:unique}(ii) by Theorem \ref{thm:unique}(iii).
\end{proof}

\section*{Proofs from Section~\ref{sec:compMOP}}

\lemEquivl*
\begin{proof}
We will consider only the first statement. The proof of the second is
analogous.

Since $\Nn$ is deterministic and the framework is Mazurkiewicz-invariant, we have $\sem{\Nn} = \sem{\Nn,S}$ for
every scheduler $S$. Let $S$ be the scheduler that gives priority to
nodes outside $\Nn|_\ell$ over nodes of $\Nn|_\ell$. 
By Theorem~\ref{thm:mazur}, every successful run $C_{\it init}
\act{w} C_{\it fin}$ of $\Nn$ compatible with $S$ can be split into $w
= w_0 u_1 w_1 \cdots u_k w_k$ such that $C_{\it init} \act{w_0}
\first{\ell}$, $\first{\ell} \act{u_i} \last{\ell} \act{w_i}
\first{\ell}$ for every $1 \leq i \leq k-1$, and $\first{\ell}
\act{u_k} \last{\ell} \act{w_k} C_{\it fin}$, where $u_1, \ldots, u_k$
are successful runs of $\Nn|_\ell$, and $w_k$ does not contain
$\ell$. 
Let ${\cal R}(w_1, \ldots, w_k)$ stand for the set of all successful
runs of this form for some fixed $w_1,\dots, w_k$.
By distributivity of $\sem{\_}$, the total
contribution of ${\cal R}(w_1, \ldots, w_k)$ to $\sem{\Nn}$ is
$\sem{{\cal R}(w_1, \ldots, w_k)}=\sem{w_0} \circ \sem{\Nn|_\ell}
\circ \sem{w_1} \circ \cdots \circ \sem{\Nn|_\ell}\circ\sem{w_k}$.

Consider now a successful run $w'$ of $\Redl(\Nn)$.
It is of the form $w_0 \ell' w_1 \ell' w_2\cdots  \ell' w_k$ with
$w_k$  not containing $\ell'=(n,a_\ell)$. 
Since $\sem{\ell'} = \sem{\Nnl}$, the contribution of $w'$ to 
$\sem{{\it Red}_\ell(\Nn)}$ is $\sem{w'} = \sem{{\cal R}(w_1, \ldots,
  w_n)}$. 
Abusing language, let us write ${\cal R}(w')$ instead of 
${\cal R}(w_1, \ldots, w_n)$. Letting $w'$ range over the successful runs of ${\it Red}_\ell(\Nn)$, we get 
$\sem{{\it Red}_\ell(\Nn)} = \sqcap_{w'} \sem{w'}  = \sqcap_{w'} \sem{{\cal R}(w')} = \sem{\Nn}$.
\end{proof}

\invariants*

\begin{proof}
\noindent (1) Let $n$ and $\ell$ be as in the assumption. By the
definition of $\Nnl$, all nodes  $m \neq n$ of $\Nnl$ (except for the
final node)
have a smaller domain than $n$ and so, by the minimality of $n$, are
reduced.
In particular, they all have a single outcome. 
Since $n$ also has one single outcome in $\Nnl$, namely $\ell$,
every node of $\Nnl$ has one single outcome. Finally, we observe that
$\Nnl$ is acyclic: any circuit $\Nnl$ would contain a dominant node
$n'\not= n$ by Lemmas~\ref{lem:lsound},
\ref{lem:domin}. Since $n'$ is reduced we have $\delta(n',a,p)=F(n')(p)$ for every $p
\in \dom(n')$. So it cannot be the case that $n'$ is dominant, by the
definition of $F(n')$.

\noindent (2) Let $C_{\it init}|_n \act{\ell_1 \cdots \ell_k} C$ be an arbitrary run of $\Nn|_n$, where $\ell_i=(n_i,a_i)$ for every
$1 \leq i \leq k$. We prove that for every $1 \leq i \leq k-1$, the outcome $\ell_i$ satisfies $\delta|_n(n_i,a,p)=n_{i+1}$ for every process $p \in \dom(n)$, which implies that $\Nn|_n$ is a replication.

Assume that the above property does not hold, and let  $\ell_i$ be the first location
that does not satisfy the property.
By the definition of $\Nn|_n$ we have $\dom(n_j) \subseteq \dom(n)$
for every $1 \leq j \leq k$.  Since $n_{i+1}$ is enabled after
the occurrence of $\ell_i$, we have $\delta(n_i, a_i, p) = n_{i+1}$
for every process $p \in \dom(n_{i+1})$. Since $\ell_i$ does not
satisfy the property, we have $\dom(n_{i+1}) \subset \dom(n_i)$. But
then $\Nn|_{\ell_i}$ contains at least the nodes $n_i$ and $n_{i+1}$,
and so the location $\ell_i$ is not reduced, contradicting the
hypothesis.
\end{proof}

%%% Local Variables:
%%% mode: latex
%%% TeX-master: "mlics"
%%% End:

%\begin{thebibliography}
%\end{thebibliography}
% that's all folks
\end{document}